\documentclass[preprint,a4paper,11pt,sort]{elsarticle}
\usepackage{fullpage}

\usepackage{graphicx}

\usepackage{amsmath}
\usepackage{amssymb}

\usepackage{amsthm}

\usepackage{hyperref}
\usepackage[capitalise,nameinlink]{cleveref}

\newtheorem{theorem}{Theorem}[section]
\newtheorem{lemma}[theorem]{Lemma}
\newtheorem{proposition}[theorem]{Proposition}
\newtheorem{corollary}[theorem]{Corollary}

\theoremstyle{definition}

\newtheorem{observation}[theorem]{Observation}

\theoremstyle{remark}
\newtheorem{myclaim}[theorem]{Claim}
\crefname{myclaim}{Claim}{Claims}

\usepackage{mathtools}
\usepackage{tcolorbox}

\usepackage{xcolor}
\definecolor{darkred}{rgb}{0.7,0,0}

\newcommand{\bigOh}{\mathcal{O}}

\newcommand{\gpara}{\mathsf}
\DeclareMathOperator{\vi}{\gpara{vi}}
\DeclareMathOperator{\li}{\gpara{li}}
\DeclareMathOperator{\wvi}{\gpara{wvi}}
\DeclareMathOperator{\weight}{\gpara{w}}
\DeclareMathOperator{\cc}{\gpara{cc}}
\DeclareMathOperator{\td}{\gpara{{td}}}
\DeclareMathOperator{\pw}{\gpara{pw}}
\DeclareMathOperator{\fvs}{\gpara{{fvs}}}
\DeclareMathOperator{\tw}{\gpara{{tw}}}
\DeclareMathOperator{\nd}{\gpara{nd}}
\DeclareMathOperator{\tc}{\gpara{tc}}
\DeclareMathOperator{\mw}{\gpara{mw}}
\DeclareMathOperator{\cw}{\gpara{cw}}
\DeclareMathOperator{\cvd}{\gpara{cvd}}
\DeclareMathOperator{\vc}{\gpara{vc}}

\newcommand{\DP}{\mathtt{dp}} 
\newcommand{\true}{\mathsf{true}}
\newcommand{\false}{\mathsf{false}}

\newcommand{\scom}{\mathtt{csum}} 
\newcommand{\mcom}{\mathtt{cmax}}

\usepackage{tikz}
\usetikzlibrary{shapes.misc}

\usepackage[mathlines]{lineno}
\newcommand*\patchAmsMathEnvironmentForLineno[1]{
  \expandafter\let\csname old#1\expandafter\endcsname\csname #1\endcsname
  \expandafter\let\csname oldend#1\expandafter\endcsname\csname end#1\endcsname
  \renewenvironment{#1}
     {\linenomath\csname old#1\endcsname}
     {\csname oldend#1\endcsname\endlinenomath}}
\newcommand*\patchBothAmsMathEnvironmentsForLineno[1]{
  \patchAmsMathEnvironmentForLineno{#1}
  \patchAmsMathEnvironmentForLineno{#1*}}
\AtBeginDocument{
\patchBothAmsMathEnvironmentsForLineno{equation}
\patchBothAmsMathEnvironmentsForLineno{align}
\patchBothAmsMathEnvironmentsForLineno{flalign}
\patchBothAmsMathEnvironmentsForLineno{alignat}
\patchBothAmsMathEnvironmentsForLineno{gather}
\patchBothAmsMathEnvironmentsForLineno{multline}
}

\newenvironment{subproof}[1][\proofname]{%
  \begin{proof}[#1]%
  
}{%
  \end{proof}%
}

\journal{Theoretical Computer Science}

\begin{document}

\title{Structural parameterizations of vertex integrity\tnoteref{t1}}
\tnotetext[t1]{%
Partially supported
by JSPS KAKENHI Grant Numbers 
JP18H04091, 
JP20H00595, 
JP20H05793, 
JP20H05967, 
JP21H05852, 
JP21K11752, 
JP21K17707, 
JP21K19765, 
JP22H00513, 
JP23H03344, 
JP23H04388, 
JP23KJ1066. 
A preliminary version appeared in the proceedings of
the 18th International Conference and Workshops on Algorithms and Computation (WALCOM 2024),
Lecture Notes in Computer Science 14549 (2024) 406--420.}

\author[1]{Tatsuya Gima} 
\ead{gima@ist.hokudai.ac.jp}

\author[2]{Tesshu Hanaka} 
\ead{hanaka@inf.kyushu-u.ac.jp}

\author[1]{Yasuaki Kobayashi} 
\ead{koba@ist.hokudai.ac.jp}

\author[3]{Ryota Murai\fnref{fn-murai}}
\ead{ryota.murai@nagoya-u.jp}

\author[3]{Hirotaka Ono} 
\ead{ono@nagoya-u.jp}

\author[3]{Yota~Otachi\corref{cor1}} 
\ead{otachi@nagoya-u.jp}

\fntext[fn-murai]{This work was done while he was a student at Nagoya University.
He is currently working at Komatsu Ltd.}

\cortext[cor1]{Corresponding author.}

\address[1]{Hokkaido University, Sapporo, Japan}
\address[2]{Kyushu University, Fukuoka, Japan}
\address[3]{Nagoya University, Nagoya, Japan}

\begin{abstract}
The graph parameter \emph{vertex integrity} measures how vulnerable a graph is to a removal of a small number of vertices.
More precisely, a graph with small vertex integrity admits a small number of vertex removals to make the remaining connected components small.
In this paper, we initiate a systematic study of structural parameterizations of the problem of computing the unweighted/weighted vertex integrity. 
As structural graph parameters, 
we consider well-known parameters such as clique-width, treewidth, pathwidth, treedepth, modular-width, neighborhood diversity, twin cover number, and cluster vertex deletion number.
We show several positive and negative results and present sharp complexity contrasts.
We also show that the vertex integrity can be approximated within an $\bigOh(\log \mathsf{opt})$ factor.
\end{abstract}

\maketitle

\section{Introduction}

Barefoot et al.~\cite{BarefootES87} introduced the concept of vertex integrity as a vulnerability measure of communication networks.
Intuitively, having small vertex integrity implies that one can remove a small set of vertices to make the remaining components small~\cite{BaggaBGLP92}.
The \emph{vertex integrity} (or just the \emph{integrity}) of a graph $G = (V,E)$, denoted $\vi(G)$, is defined as
\[
  \vi(G)=\min_{S\subseteq V} \Big\{|S|+\max_{C\in\cc(G-S)}|V(C)|\Big\},
\]
where $G - S$ denotes the graph obtained from $G$ by deleting all vertices in $S$, and $\cc(G-S)$ is the set of connected components of $G - S$.
For a vertex-weighted graph $G = (V,E,\weight)$ with weight function $\weight \colon V \to \mathbb{Z}^{+}$,
we can also define the \emph{weighted vertex integrity} of $G$, denoted $\wvi(G)$,
by replacing $|S|$ and $|V(C)|$ with $\weight(S)$ and $\weight(V(C))$ in the definition, respectively,
where $\weight(X) = \sum_{v \in X} \weight(v)$ for $X \subseteq V$.\footnote{%
We consider positive weights only since a vertex of non-positive weight is safely removed from the graph.}
Note that the unweighted vertex integrity $\vi(G)$ can be still defined for a vertex-weighted graph $G = (V,E,\weight)$ by just ignoring $\weight$
(or using a unit-weight function).
For a weighted graph $G = (V,E,\weight)$,
a set $S \subseteq V$ is a \emph{$\wvi(k)$-set} if
\[
 \weight(S) + \max_{C \in \cc(G-S)} \weight(V(C)) \le k.
\]
A $\wvi(k)$-set is a \emph{$\wvi$-set} if $k = \wvi(G)$.
In the analogous ways, $\vi(k)$-set and $\vi$-set are defined.

In this paper, we study the problems of computing the unweighted/weighted vertex integrity of a graph, which can be formalized as follows:
\begin{tcolorbox}
\begin{description}
  \setlength{\itemsep}{0pt}
  \item[Problem.] \textsc{Unweighted/Weighted Vertex Integrity}
  \item[Input.] A graph $G$ (with $\weight \colon V \to \mathbb{Z}^{+}$ in the weighted version) and an integer $k$.
  \item[Question.] Is $\vi(G) \le k$?/Is $\wvi(G) \le k$?
\end{description}
\end{tcolorbox}
\noindent 
Note that the complexity of \textsc{Weighted Vertex Integrity}
may depend on the representation of the vertex-weight function $\weight$.
We denote by \textsc{Binary/Unary Weighted Vertex Integrity}
the two cases where $\weight$ is encoded in binary and unary, respectively.

\textsc{Unweighted Vertex Integrity} has been studied on several graph classes.
It is NP-complete on planar graphs~\cite{ClarkEF1987}, and on co-bipartite graphs and chordal graphs~\cite{DrangeDH16}.
On the other hand, the problem becomes tractable when the input is restricted to some classes of graphs:
trees and cactus graphs~\cite[without proofs]{BaggaBGLP92},
graphs with linear structures (such as interval graphs, circular-arc graphs, permutation graphs, trapezoid graphs, and co-comparability graphs of bounded dimension)~\cite{KratschKM97},
and on split graphs \cite{LiZZ08}.
The parameterized complexity of \textsc{Unweighted Vertex Integrity} with the natural parameter $k$ was first addressed by Fellows and Stueckle~\cite{FellowsS89},
who showed that it can be solved in $\bigOh(k^{3k}n)$ time. 
Drange et al.~\cite{DrangeDH16} later generalized and improved the result by presenting an algorithm for \textsc{Weighted Vertex Integrity}
with the running time $\bigOh(k^{k+1}n)$.
The existence of an $\bigOh(c^{k} n^{\bigOh(1)})$-time algorithm for a constant $c$ is open even in the unweighted case.
Drange et al.~\cite{DrangeDH16} also presented a $k^{3}$-vertex kernel for \textsc{Weighted Vertex Integrity}.
Recently, the kernel size has been improved by Casel et al.~\cite{CaselFNSZ24}.
In particular, they showed that \textsc{Unweighted Vertex Integrity} admits a $k^{2}$-vertex kernel.

The $\bigOh(k^{k+1}n)$-time algorithm for \textsc{Weighted Vertex Integrity} by Drange et al.~\cite{DrangeDH16}
implies that \textsc{Weighted Vertex Integrity} is fixed-parameter tractable parameterized by weighted vertex integrity
and \textsc{Unweighted Vertex Integrity} is fixed-parameter tractable parameterized by unweighted vertex integrity.
To the best of our knowledge, there has been no other result dealing with structural parameterizations of the \textsc{Unweighted/Weighted Vertex Integrity}.

Recently, vertex integrity has also been studied as a structural graph parameter since 
it is an upper bound of treedepth and a lower bound of vertex cover number plus~$1$.
This line of research has been studied extensively~\cite{GimaHKKO22,LampisM21,GimaO24,BentertHK23}.


\subsection{Our results}

In this paper, we study \textsc{Unweighted/Weighted Vertex Integrity} parameterized by well-studied structural parameters and show sharp complexity contrasts.
We also study the problem on important graph classes (e.g., bipartite planar graphs) for which the complexity of the problem was unknown.
Our results can be summarized as follows (see \cref{fig:graph-parameters,fig:graph-classes}):\\
(1) \textsc{Unweighted Vertex Integrity} is
\begin{itemize}
  \item FPT parameterized by cluster vertex deletion number (\cref{thm:unweighted_cvd}),
  \item W[2]-hard parameterized by pathwidth (\cref{cor:unweighted_pw}),
  \item NP-complete on planar bipartite graphs of maximum degree~4 (\cref{thm:planar-bipartite}),
  \item NP-complete on line graphs (\cref{cor:line-graph});
\end{itemize}
(2) \textsc{Unary Weighted Vertex Integrity} is
\begin{itemize}
  \item FPT parameterized by modular-width (\cref{thm:unary_mw}),
  \item W[1]-hard parameterized 
  by cluster vertex deletion number and unweighted vertex integrity
  or by feedback vertex set number and unweighted vertex integrity (\cref{thm:unary-vi-cvd-fvs_Wh}),
  \item XP parameterized by clique-width (\cref{thm:unary-cw-xp})
  (and thus polynomial-time solvable on distance hereditary graphs);
\end{itemize}
(3) \textsc{Binary Weighted Vertex Integrity} is
\begin{itemize}
  \item FPT parameterized by neighborhood diversity and by twin cover number (\cref{cor:bin-nd,cor:bin-tc}),
  \item NP-complete on subdivided stars 
  (and thus paraNP-complete parameterized simultaneously by cluster vertex deletion number, unweighted vertex integrity, and feedback vertex set number) 
  (\cref{thm:binary_subdivided_star_NPc}).
\end{itemize}
Additionally, we show that by slightly modifying the approximation algorithm by Lee~\cite[Theorem~1]{Lee19} for a related problem as a subroutine,
the weighted vertex integrity can be approximated within an $\bigOh(\log \mathsf{opt})$ factor (\cref{thm:log-approx}).

Since we focus on the classification of the parameterized complexity with respect to different parameters, 
we do not optimize or specify the running time of our algorithms.

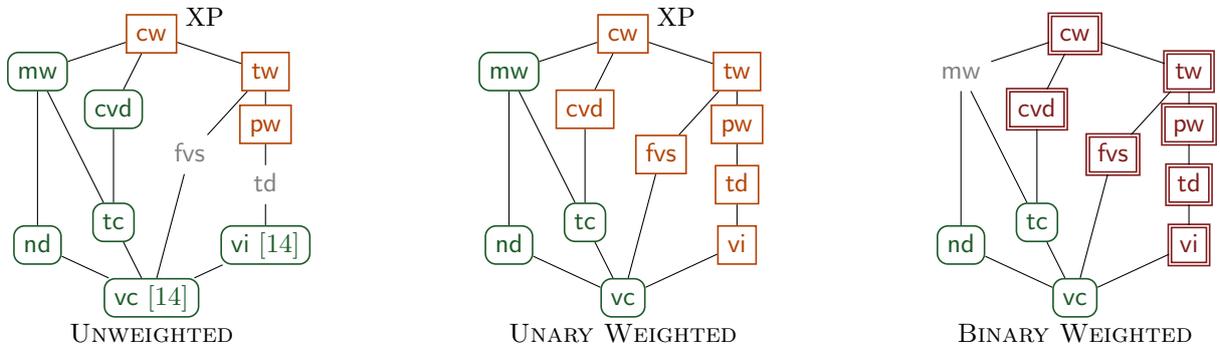
\begin{figure}[ht]
    \centering

    \definecolor{tkzdarkred}{rgb}{.5,.1,.1}
    \definecolor{tkzdarkorange}{rgb}{.7,.25,0}
    \definecolor{tkzdarkblue}{rgb}{.1,.1,.5}
    \definecolor{tkzdarkgreen}{rgb}{.1,.35,.15}

    \tikzset{npc/.style = {draw,semithick,rectangle,double,tkzdarkred,align=center}}
    \tikzset{whard/.style = {draw,semithick,rectangle,tkzdarkorange,align=center}}
    \tikzset{fpt/.style = {draw,semithick,rectangle,rounded corners,tkzdarkgreen,align=center}}
    \tikzset{unknown/.style = {gray,align=center}}

    \begin{tikzpicture}[every node/.style={font=\small,text height=1.3ex, text depth = 0.2ex},scale=1.0]
        \node[whard] (cw) at (0cm, 4cm) {$\cw$};
        \node[fpt] (mw) at (-1.5cm, 3.5cm) {$\mw$};
        \node[fpt] (nd) at (-1.5cm, 1.2cm) {$\nd$};
        \node[fpt] (cvd) at (-.5cm, 3cm) {$\cvd$};
        \node[fpt] (tc) at (-.5cm, 1.5cm) {$\tc$};
        \node[unknown] (fvs) at (.5cm, 2.4cm) {$\fvs$};
        \node[whard] (tw) at (1.5cm, 3.5cm) {$\tw$};
        \node[whard] (pw) at (1.5cm, 2.8cm) {$\pw$};
        \node[unknown] (td) at (1.5cm, 2.0cm) {$\td$};
        \node[fpt] (vi) at (1.5cm, 1.2cm) {$\vi$ \cite{DrangeDH16}};
        \node[fpt] (vc) at (0cm, 0.5cm) {$\vc$ \cite{DrangeDH16}};

        \draw (cw) -- (mw) -- (nd) -- (vc);
        \draw (cw) -- (cvd) -- (tc) -- (vc);
	\draw (mw) -- (tc);
        \draw (cw) -- (tw) -- (pw) -- (td) -- (vi) -- (vc);
        \draw (tw) -- (fvs) -- (vc);

	\node (caption) at (0cm, 0cm) {\textsc{Unweighted}};
	\node (xp) at (.7cm, 4.2cm) {XP};
    \end{tikzpicture} 
    \hfill
    \begin{tikzpicture}[every node/.style={font=\small,text height=1.3ex, text depth = 0.2ex},scale=1.0]
        \node[whard] (cw) at (0cm, 4cm) {$\cw$};
        \node[fpt] (mw) at (-1.5cm, 3.5cm) {$\mw$};
        \node[fpt] (nd) at (-1.5cm, 1.2cm) {$\nd$};
        \node[whard] (cvd) at (-.5cm, 3cm) {$\cvd$};
        \node[fpt] (tc) at (-.5cm, 1.5cm) {$\tc$};
        \node[whard] (fvs) at (.5cm, 2.4cm) {$\fvs$};
        \node[whard] (tw) at (1.5cm, 3.5cm) {$\tw$};
        \node[whard] (pw) at (1.5cm, 2.8cm) {$\pw$};
        \node[whard] (td) at (1.5cm, 2.0cm) {$\td$};
        \node[whard] (vi) at (1.5cm, 1.2cm) {$\vi$};
        \node[fpt] (vc) at (0cm, 0.5cm) {$\vc$};

        \draw (cw) -- (mw) -- (nd) -- (vc);
        \draw (cw) -- (cvd) -- (tc) -- (vc);
	\draw (mw) -- (tc);
        \draw (cw) -- (tw) -- (pw) -- (td) -- (vi) -- (vc);
        \draw (tw) -- (fvs) -- (vc);

	\node (caption) at (0cm, 0cm) {\textsc{Unary Weighted}};
	\node (xp) at (.7cm, 4.2cm) {XP};
    \end{tikzpicture} 
    \hfill
    \begin{tikzpicture}[every node/.style={font=\small,text height=1.3ex, text depth = 0.2ex},scale=1.0]
        \node[npc] (cw) at (0cm, 4cm) {$\cw$};
        \node[unknown] (mw) at (-1.5cm, 3.5cm) {$\mw$};
        \node[fpt] (nd) at (-1.5cm, 1.2cm) {$\nd$};
        \node[npc] (cvd) at (-.5cm, 3cm) {$\cvd$};
        \node[fpt] (tc) at (-.5cm, 1.5cm) {$\tc$};
        \node[npc] (fvs) at (.5cm, 2.4cm) {$\fvs$};
        \node[npc] (tw) at (1.5cm, 3.5cm) {$\tw$};
        \node[npc] (pw) at (1.5cm, 2.8cm) {$\pw$};
        \node[npc] (td) at (1.5cm, 2.0cm) {$\td$};
        \node[npc] (vi) at (1.5cm, 1.2cm) {$\vi$};
        \node[fpt] (vc) at (0cm, 0.5cm) {$\vc$};

        \draw (cw) -- (mw) -- (nd) -- (vc);
        \draw (cw) -- (cvd) -- (tc) -- (vc);
	\draw (mw) -- (tc);
        \draw (cw) -- (tw) -- (pw) -- (td) -- (vi) -- (vc);
        \draw (tw) -- (fvs) -- (vc);

	\node (caption) at (0cm, 0cm) {\textsc{Binary Weighted}};
    \end{tikzpicture} 

    \caption{The complexity of \textsc{Unweighted/Weighted Vertex Integrity} with structural parameters.
      (See \cref{sec:pre} for the definitions of the acronyms.)
      The results without references are shown in this paper. The double, single, and rounded rectangles indicate
      paraNP-complete, W[1]-/W[2]-hard, and fixed-parameter tractable cases, respectively.
      A connection between two parameters implies that the one above generalizes the one below; that is, one below is lower-bounded by a function of the one above.}
    \label{fig:graph-parameters}
\end{figure}

\begin{figure}[ht]
    \centering

    \definecolor{tkzdarkred}{rgb}{.5,.1,.1}
    \definecolor{tkzdarkblue}{rgb}{.1,.1,.5}

    \tikzset{npc/.style = {draw,semithick,rectangle,double,tkzdarkred,align=center}}
    \tikzset{poly/.style = {draw,semithick,rectangle,rounded corners,tkzdarkblue,align=center}}
    \tikzset{unknown/.style = {gray,align=center}}

    \begin{tikzpicture}[every node/.style={font=\small},scale=1.1]
        \node[npc]  (pl) at (0cm, 4cm) {planar \cite{ClarkEF1987}};
        \node[npc] (bpl) at (0cm, 1.2cm) {bipartite $\cap$ \\ planar $\cap$ \\ $\max \deg \le 4$*};
        \node[npc]  (bi) at (1.5cm, 2.5cm) {bipartite};
        \node[poly]  (bp) at (1.5cm, 0cm) {bipartite permutation};
        \node[poly]  (pm) at (4cm, 0.9cm) {permutation};
        \node[poly]  (co) at (4.2cm, 0cm) {cograph};
        \node[npc]  (comp) at (2cm, 4cm) {comparability};
        \node[poly]  (trap) at (4.5cm, 1.7cm) {$d$-trapezoid};
        \node[poly]  (ccd) at (4.7cm, 3cm) {co-comparability \\ of bounded \\ dimension \cite{KratschKM97}};
        \node[npc]  (cc) at (4.7cm, 4.3cm) {co-comparability};
        \node[npc]  (atf) at (4.7cm, 5.2cm) {AT-free};
        
        \node[npc]  (ch) at (10.5cm, 4cm) {chordal~\cite{DrangeDH16}};
        \node[poly]  (int) at (8.7cm, 2.9cm) {interval};
        \node[poly]  (ca) at (8.7cm, 5cm) {circular-arc \\ \cite{KratschKM97}};
        \node[npc] (cb) at (10.5cm, 0cm) {co-bipartite \\ \cite{DrangeDH16}};
        \node[npc]  (wch) at (6.7cm, 5cm) {weakly \\ chordal};
        \node[poly] (pig) at (8.7cm, 0cm) {proper \\ interval};
        \node[npc]  (cf) at (10.5cm, 2cm) {claw-free};
        \node[npc]  (line) at (12.2cm, 0cm) {line*};
        \node[unknown]  (cir) at (7.3cm, 2.7cm) {circle};
        \node[poly] (dh) at (7.3cm, 1.1cm) {distance \\ hereditary*};
        \node[poly] (sp) at (11.8cm, 3cm) {split \cite{LiZZ08}};
        
        \draw (bi) -- (bpl);
        \draw (pl) -- (bpl);
        \draw (bi) -- (bp);
        \draw (pm) -- (bp);
        \draw (comp) -- (bi);
        \draw (comp) -- (pm);
        \draw (atf) -- (cc);
        \draw (cc) -- (ccd);
        \draw (ccd) -- (trap);
        \draw (trap) -- (pm);
        \draw (cc) [out=360] to (int);
        \draw (ch) -- (int);
        \draw (wch) [out=320,in=180] to (ch);
        \draw (wch) [out=270,in=120] to (dh);
        \draw (ca) -- (int);
        \draw (cc) [out=350,in=140] to (cb); 
        \draw (int) -- (pig);
        \draw (cf) -- (pig);
        \draw (cf) -- (line);
        \draw (cf) -- (cb);
        \draw (cir) -- (dh);
        \draw (ch) -- (sp);
        \draw (cir) [out=190,in=5] to (pm); 
        \draw (cir) [out=330,in=110] to (pig); 
        \draw (dh) [out=190,in=10] to (co); 
        \draw (pm) -- (co); 
    \end{tikzpicture} 
    \caption{The complexity of \textsc{Unweighted Vertex Integrity} on graph classes.
      (See \cite{BrandstadtLS99} for the definitions of the graph classes.)
      The ones marked with asterisks are shown in this paper.
      The double and rounded rectangles indicate
      NP-complete and polynomial-time solvable cases, respectively.
      A connection between two graph classes indicates that the one above is a superclass of the one below.
    }
    \label{fig:graph-classes}
\end{figure}
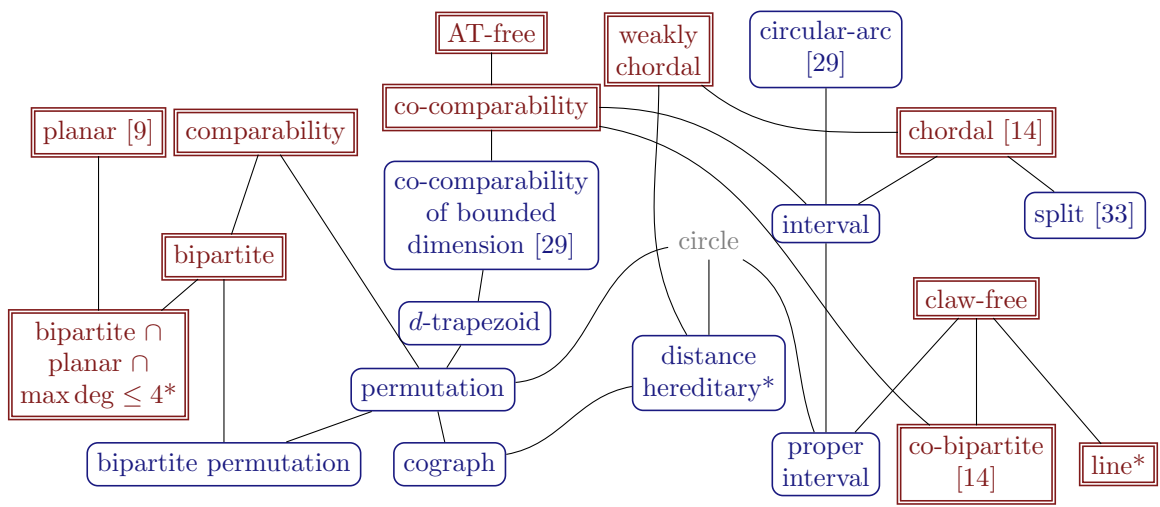

\subsection{Overview of the proof ideas}
\label{sec:proof-ideas}
To get a quick overview of the ideas of the proofs,
let us briefly discuss how we obtain the results in this paper.

\paragraph{Positive results}
The fixed-parameter algorithms for the binary case parameterized by neighborhood diversity and by twin cover number (\cref{cor:bin-nd,cor:bin-tc})
are based on a simple observation that there is a $\vi$-set $S$ such that $S \cap T \in \{\emptyset,T\}$ holds for each twin class $T$.
The fixed-parameter algorithm for the unary case parameterized by modular-width (\cref{thm:unary_mw}) is based on a similar idea but much more involved
as it needs to recursively consider the intersections with a $\vi$-set and modules.
The XP algorithm for the unary case parameterized by clique-width (\cref{thm:unary-cw-xp}) is a standard dynamic programming.

The fixed-parameter algorithm for the unweighted case parameterized by cluster vertex deletion number (\cref{thm:unweighted_cvd})
is actually the most technically involved result in this paper.
It needs a careful analysis on how a $\vi$-set intersects the cliques and then it reduces the problem to a subproblem.
Solving the subproblem is still nontrivial, but we can do it by showing that a greedy strategy works in a key step.

\paragraph{Negative results}
The W[2]-hardness of the unweighted case parameterized by pathwidth (\cref{cor:unweighted_pw}) is shown by a reduction from a similar problem.
For the unary case, a reduction from \textsc{Unary Bin Packing} shows the W[1]-hardness parameterized by cluster vertex deletion number and unweighted vertex integrity
or by feedback vertex set number and unweighted vertex integrity (\cref{thm:unary-vi-cvd-fvs_Wh}).
A similar reduction shows the NP-completeness of the unweighted case on line graphs (\cref{cor:line-graph}) as well.
The NP-completeness of the binary case on subdivided stars (\cref{thm:binary_subdivided_star_NPc}) is proved by a reduction from \textsc{Partition}.
A reduction from \textsc{Vertex Cover} shows that the unweighted case is NP-complete on planar bipartite graphs of maximum degree~4 (\cref{thm:planar-bipartite}),

\subsection{Related graph parameters}
\label{sec:related-paras}
Since the concept of vertex integrity is natural,
there are a few other parameters defined in similar ways.
The \emph{fracture number}~\cite{DvorakEGKO21} of a graph 
is the minimum $k$ such that one can remove at most $k$ vertices to make the maximum size of a remaining component at most~$k$.
The \emph{starwidth}~\cite{Ee17} is the minimum width of a tree decomposition restricted to be a star.
The \emph{safe number}~\cite{FujitaF18} of $G$ is the minimum size of a non-empty vertex set $S$ such that each connected component in $G[S]$
is not smaller than any adjacent connected component in $G - S$.
As structural graph parameters, these parameters including vertex integrity are equivalent in the sense that one is small if and only if the others are small.
More precisely, the starwidth, the fracture number, and the vertex integrity of a graph differ only by a constant factor,
while the safe number may be as small as the square root of the vertex integrity.
Note however that computing them exactly would be quite different tasks.
For example, computing the safe number is NP-hard on split graphs~\cite{AguedaCFLMMMNOO18},
while computing the unweighted vertex integrity is polynomial-time solvable on them~\cite{LiZZ08}.
Furthermore, in the binary-weighted setting, we can see that computing the fracture number or the safe number is NP-hard on complete graphs
as it generalizes the classical NP-complete problem \textsc{Partition}~\cite{GareyJ79},
while the computation of the vertex integrity becomes trivial as it is equal to the sum of the vertex weights in this case.

The \emph{$\ell$-component order connectivity} is the minimum size $p$ of a vertex set such that 
the removal of the vertex set makes the maximum order of a remaining connected component at most~$\ell$~\cite{DrangeDH16}.
Although we cannot directly compare this parameter to vertex integrity,
some of the techniques used in this paper would be useful for studying $\ell$-component order connectivity as well.
Note that the general result by Cai~\cite{Cai96} implies that
for fixed $\ell$, computing $\ell$-component order connectivity is fixed-parameter parameterized by $p$.
This also applies to the cluster vertex deletion number, which is defined later in \cref{ssec:cvd}.
(We mention a specific faster algorithm there.)


\section{Preliminaries}
\label{sec:pre}
We assume that the readers are familiar with the theory of fixed-parameter tractability.
See standard textbooks (e.g., \cite{CyganFKLMPPS15,DowneyF13,FlumG06,Niedermeier06}) for the definitions omitted in this paper.

For a graph $G$, we denote
its clique-width by $\cw(G)$,
treewidth by $\tw(G)$,
pathwidth by $\pw(G)$,
treedepth by $\td(G)$,
vertex cover number by $\vc(G)$,
feedback vertex set number by $\fvs(G)$,
modular-width by $\mw(G)$,
neighborhood diversity by $\nd(G)$,
cluster vertex deletion set number by $\cvd(G)$, and
twin cover number by $\tc(G)$.
We defer their definitions until needed and if a full definition is not necessary, 
we omit it and present only the properties needed.
(See, e.g., \cite{GajarskyLO13,HlinenyOSG08,SorgeW19} for their definitions.)
Note that for a vertex-weighted graph, we define these structural parameters on the underlying unweighted graph.

\paragraph{Irredundant set}
Let $G$ be a graph and $S \subseteq V(G)$.
A vertex $v \in S$ is \emph{redundant} if at most one connected component of $G - S$ contains neighbors of $v$.
The set $S$ is \emph{irredundant} if it contains no redundant vertices.
A vertex is \emph{simplicial} if its neighborhood is a clique.
Since no vertex set can separate a clique into two or more connected components,
simplicial vertices cannot belong to irredundant sets.
\begin{observation}
\label{obs:simplicial}
An irredundant set contains no simplicial vertex.
\end{observation}
The following is shown in \cite[Lemma~3.1]{DrangeDH16} for simplicial vertices,
but their proof works for this generalization to redundant vertices.
\begin{lemma}
[{\cite[Lemma~3.1]{DrangeDH16}}]
\label{lem:irredundant}
Let $G = (V,E,\weight)$ be a vertex-weighted graph.
If $S \subseteq V$ contains a redundant vertex $v$, then
\[
  \weight(S \setminus \{v\}) + \max_{C \in \cc(G-(S\setminus \{v\}))} \weight(V(C)) \le \weight(S) + \max_{C \in \cc(G-S)} \weight(V(C)).
\]
\end{lemma}
\begin{proof}
If no connected component of $G-S$ contains neighbors of $v$,
then $\cc(G-(S \setminus \{v\}))$ consists of all connected components in $\cc(G-S)$ and the one containing only $v$.
If there exists exactly one connected component $D$ of $G-S$ that contains neighbors of $v$,
then $\cc(G-(S \setminus \{v\}))$ consists of $\cc(G-S) \setminus \{D\}$ and the connected component induced by $V(D) \cup \{v\}$.
In both cases, we have
\[
  \max_{C \in \cc(G-(S\setminus \{v\}))} \weight(V(C)) \le \weight(v) + \max_{C \in \cc(G-S)} \weight(V(C)).
\]
Since $\weight(S \setminus \{v\}) + \weight(v) = \weight(S)$, the lemma follows.
\end{proof}
\begin{corollary}
\label{cor:irredundant}
A graph with a $\wvi(k)$-set has an irredundant $\wvi(k)$-set.
\end{corollary}

We denote by $N(v)$ and $N[v]$ the (\emph{open}) \emph{neighborhood} and the \emph{closed neighborhood} of $v$, respectively.
For a set $S$ of vertices, we define $N(S) = \bigcup_{v \in S} N(v) \setminus S$.
Two vertices $u,v \in V(G)$ are \emph{twins} if they have the same neighborhood except for themselves;
that is, if $N(u) \setminus \{v\} = N(v) \setminus \{u\}$ holds.
A \emph{twin class} $T$ of $G$ is a maximal set of twin vertices in $G$.
We can show that an irredundant set contains either none or all of a twin class.
\begin{lemma}
\label{lem:twin-all/nothing}
Let $G$ be a graph.
If $S \subseteq V(G)$ is an irredundant set,
then for each twin class $T$ of $G$
it holds that $T \cap S = \emptyset$ or $T \subseteq S$.
\end{lemma}
\begin{proof}
Suppose to the contrary that there are two vertices $u \in T \cap S$ and $v \in T \setminus S$.
Let $C$ be the connected component of $G - S$ that contains $v$.
Observe that $N[v] \setminus S \subseteq V(C)$.
Since $N(u) \setminus \{v\} = N(v) \setminus \{u\}$,
we have $N(u) \subseteq N[v]$, and thus
$N(u) \setminus S \subseteq N[v] \setminus S \subseteq V(C)$.
This implies that $u$ is redundant.
\end{proof}


\section{Positive results}

Here we present our algorithmic results.
In the descriptions of algorithms, we sometimes use a phrase like ``to guess something.''
For example, when an integer $\ell$ belongs to $\{1,\dots,c\}$,
we may say that ``we guess $\ell$ from $\{1,\dots,c\}$.''
This means that we try all possibilities $\ell = 1, \dots, c$,
and output the optimal one, where the optimality depends on (and should be clear from) the context.
We thus need a multiplicative factor of $c$ in the running time.
After ``guessing'' some object, we assume that the object is fixed to one of the candidates.

\subsection{The binary-weighted problem parameterized by $\nd$ or $\tc$}
We first observe that the most general case \textsc{Binary Weighted Vertex Integrity}
is fixed-parameter tractable when parameterized by generalizations of vertex cover number:
neighborhood diversity~\cite{Lampis12} and twin cover number~\cite{Ganian11}.

The \emph{neighborhood diversity} of a graph $G$, denoted $\nd(G)$, is the number of twin classes in $G$.
The following is an immediate corollary of \cref{lem:twin-all/nothing}.
\begin{corollary}
\label{cor:bin-nd}
\textsc{Binary Weighted Vertex Integrity} is fixed-parameter tractable parameterized by neighborhood diversity.
\end{corollary}
\begin{proof}
Let $G = (V,E)$ be a graph.
Let $S \subseteq V$ be an irredundant set.
By \cref{lem:twin-all/nothing}, each twin class is completely contained in $S$ or has no intersection with $S$.
Therefore, we only need to try $2^{\nd(G)}$ possibilities of $S$.
\end{proof}

The \emph{twin cover number} of a graph $G$, denoted $\tc(G)$, is the minimum size of a set $C \subseteq V(G)$ such that
each connected component of $G - C$ is a set of twin vertices in $G$ (and thus a clique).
\cref{lem:twin-all/nothing,cor:bin-nd} together imply the following.
\begin{corollary}
\label{cor:bin-tc}
\textsc{Binary Weighted Vertex Integrity} is fixed-parameter tractable parameterized by twin cover number.
\end{corollary}
\begin{proof}
Let $C$ be a twin cover of $G$.
By \cref{lem:twin-all/nothing}, we can replace each maximal set $K \subseteq V(G-C)$ of twin vertices by a single vertex of weight $w(K)$.
After that, $C$ becomes a vertex cover of the resultant graph, and thus the resultant graph has neighborhood diversity at most $2^{|C|} + |C|$.
Now we can apply \cref{cor:bin-nd}.
\end{proof}

\subsection{The unary-weighted problem parameterized by $\mw$}
Next we consider modular-width~\cite{GajarskyLO13}, which is a generalization of both neighborhood diversity and twin cover number, 
and show that \textsc{Unary Weighted Vertex Integrity} is fixed-parameter tractable with this parameter.

Let $H$ be a graph with two or more vertices $v_{1}, \dots, v_{c}$, and let $H_{1}, \dots, H_{c}$ be $c$ disjoint graphs.
The \emph{substitution} $H(H_{1}, \dots, H_{c})$ of the vertices of $H$ by $H_{1}, \dots, H_{c}$ is the graph
with
\begin{align*}
  V(H(H_{1}, \dots, H_{c}))
  &= \bigcup_{1 \le i \le c} V(H_{i}), 
  \\
  E(H(H_{1}, \dots, H_{c})) &= \bigcup_{1 \le i \le c} E(H_{i})
  \cup \bigcup_{\{v_{i}, v_{j}\} \in E(H)} \{\{u, w\} \mid u \in V(H_{i}), w \in V(H_{j})\}.
\end{align*}
That is,
$H(H_{1}, \dots, H_{c})$ is obtained from the disjoint union of $H_{1}, \dots, H_{c}$
by adding all possible edges between $V(H_{i})$ and $V(H_{j})$ for each edge $\{v_{i}, v_{j}\}$ of $H$.
Each $V(H_{i})$ is a \emph{module} of $H(H_{1}, \dots, H_{c})$.

A \emph{modular decomposition} is a rooted ordered tree $T$ such that 
each non-leaf node with $c$ children is labeled with a graph of $c$ vertices
and each node represents a graph as follows:
\begin{itemize}
  \item a leaf node represents the one-vertex graph;
  \item a non-leaf node labeled with a graph $H$ with $c$ vertices $v_{1}, \dots, v_{c}$
  (and thus with exactly $c$ children)
  represents the graph $H(H_{1}, \dots, H_{c})$, where $H_{i}$ is the graph represented by the $i$th child.
\end{itemize}
A \emph{modular decomposition of a graph $G$} is a modular decomposition whose root represents a graph isomorphic to~$G$.
The \emph{width} of a modular decomposition is the maximum number of children of an inner node.
The \emph{modular-width} of a graph $G$, denoted $\mw(G)$, is the minimum width of a modular decomposition of $G$.
It is known that a modular decomposition of the minimum width can be computed in linear time (and thus has a liner number of nodes)~\cite{McConnellS99}.

\begin{theorem}
\label{thm:unary_mw}
\textsc{Unary Weighted Vertex Integrity} is fixed-parameter tractable parameterized by modular-width.
\end{theorem}
\begin{proof}
Let $G = (V,E, \weight)$ be a vertex-weighted graph.
We guess an integer $\ell \in \{1,\dots,\weight(V)\}$ such that $G$ has an irredundant $\wvi$-set $S$
such that
\[
 \max_{C \in \cc(G-S)} \weight(V(C)) = \ell.
\]
As there are only $\weight(V)$ candidates of $\ell$, which is polynomial in the input size,
we can assume that $\ell$ is correctly guessed.\footnote{%
Note that this is the only part that requires the unary representation of weights.
Note also that we cannot binary-search $\ell$ as the irredundancy makes the problem non-monotone.}

We compute a modular decomposition of $G$ with width $\mw(G)$ in linear time~\cite{McConnellS99} and then proceed in a bottom-up manner.
For each graph $G'$ represented by a node in the modular decomposition,
we compute the minimum weight of a set $S$ irredundant in $G'$ such that $\max_{C \in \cc(G'-S)} \weight(V(C)) \le \ell$.
Let $\mu(G')$ denote the minimum weight of such $S$.

We first consider the case where $G'$ has only one vertex $v$.
If $\weight(v) \le \ell$, then we set $\mu(G') = 0$ as the empty set is irredundant.
Otherwise we set $\mu(G') = \infty$ since the entire vertex set $\{v\}$ is redundant.

In the following, we consider the case of $G' = H(H_{1}, \dots, H_{c})$.
For simplicity, we identify the vertices $v_{1}, \dots, v_{c}$ of $H$ with the integers $1, \dots, c$.
For a hypothetical irredundant set $S$ that we are looking for,
we guess a partition $(I_{\textrm{f}}, I_{\textrm{p}}, I_{\emptyset})$ of $\{1, \dots, c\}$ such that
\begin{itemize}
  \item if $i \in I_{\textrm{f}}$, then $V(H_{i}) \subseteq S$;
  \item if $i \in I_{\textrm{p}}$, then $V(H_{i}) \cap S \ne \emptyset$ and $V(H_{i}) \not\subseteq S$;
  \item if $i \in I_{\emptyset}$, then $V(H_{i}) \cap S = \emptyset$.
\end{itemize}
As there are only $3^{c} \le 3^{\mw(G)}$ candidates for the partition $(I_{\textrm{f}}, I_{\textrm{p}}, I_{\emptyset})$,
we assume that we have correctly guessed it.
For each connected component $C$ of $H - I_{\textrm{f}}$, 
we compute the minimum weight $\mu(C)$ of vertices that we need to remove from 
the modules of $G'$ corresponding to $C$.

\begin{myclaim}
\label{clm:isolated}
Each $i \in I_{\textrm{p}}$ has degree~$0$ in $H - I_{\textrm{f}}$.
\end{myclaim}
\begin{subproof}
[Proof of \cref{clm:isolated}]
Let $C_{i}$ be the connected component of $H - I_{\textrm{f}}$ that contains $i$.
Suppose to the contrary that $|V(C_{i})| \ge 2$.
Observe that for each edge $\{j, h\} \in E(C_{i})$, 
the subgraph of $G'$ induced by $(V(H_{j}) \cup V(H_{h})) \setminus S$ is connected
since 
both $V(H_{j}) \setminus S$ and $V(H_{h}) \setminus S$ are nonempty as $j, h \notin I_{\textrm{f}}$
and there are all possible edges between them as $\{j, h\} \in E(C_{i}) \subseteq E(H)$.
This fact and the connectivity of $C_{i}$ imply that 
$(\bigcup_{j \in V(C_{i})} V(H_{j})) \setminus S$ induces a connected component $D$ of $G' - S$.

Since $i \in I_{\textrm{p}}$, there are vertices $u \in V(H_{i}) \cap S$ and $v \in V(H_{i}) \setminus S$ ($\subseteq V(D)$).
Let $w$ be a neighbor of $u$ such that $w \notin S$.
If $w \in V(H_{i})$, then $w$ belongs to $D$.
If $w \notin V(H_{i})$, then $v$ is also adjacent to $w$ since $H_{i}$ is a module.
This implies that $w$ belongs to $D$ in this case as well.
Hence, no connected component of $G' - S$ other than $D$ may contain a neighbor of $u$.
This implies that $u$ is redundant, a contradiction.
\end{subproof}

\cref{clm:isolated} implies that a component $C$ of $H - I_{\textrm{f}}$ is either
a singleton formed by some $i \in I_{\textrm{p}}$, or one formed by a subset of $I_{\emptyset}$.
In the second case, $\mu(C) = 0$ if the total weight of the vertices in the corresponding modules is at most $\ell$;
and $\mu(C) = \infty$ otherwise.
For the first case, let $i \in I_{\textrm{p}}$ be the vertex forming $C$.
By \cref{clm:isolated}, every neighbor $j$ of $i$ in $H$ satisfies $j \in I_{\textrm{f}}$ (if any exists).
Thus, $S \cap V(H_{i})$ has to be irredundant for $H_{i}$.
This implies that $\mu(C) = \mu(H_{i})$ as we can consider $H_{i}$ independently from the rest of the graph.
Thus we can compute $\mu(G')$ as
\[
  \mu(G') 
  = \sum_{i \in I_{\textrm{f}}} \weight(V(H_{i}))
  + \sum_{i \in I_{\textrm{p}}} \mu(H_{i}).
\]
Since $\mu(H_{i})$ for every $i \in I_{\textrm{p}}$ is already computed in the lower layers of the bottom-up computation,
the computation of $\mu(G')$ can be done in polynomial time. This completes the proof.
\end{proof}

\subsection{The unweighted problem parameterized by $\cvd$}
\label{ssec:cvd}
The \emph{cluster vertex deletion number}~\cite{DouchaK12} of a graph $G$, denoted $\cvd(G)$,
is the minimum size of a \emph{cluster vertex deletion set};
that is, a set of vertices whose removal makes the remaining graph a disjoint union of complete graphs.
Finding a minimum cluster vertex deletion set is fixed-parameter tractable parameterized by 
$\cvd(G)$~\cite{HuffnerKMN10}, and thus we assume that such a set is given when $\cvd(G)$ is part of the parameter.
Note that the definition directly implies that $\cvd(G) \le \tc(G)$. 
By considering subdivided stars and complete bipartite graphs, one can see that 
cluster vertex deletion number is incomparable with neighborhood diversity and modular-width.

We show that \textsc{Unweighted Vertex Integrity} is fixed-parameter tractable parameterized by cluster vertex deletion number.
This result provides an interesting contrast with the W[1]-hardness (\cref{thm:unary-vi-cvd-fvs_Wh}) and NP-completeness (\cref{cor:binary-vi-cvd-fvs_NPc}) of
\textsc{Unary/Binary Weighted Vertex Integrity}, respectively, parameterized by the same parameter.

\begin{theorem}
\label{thm:unweighted_cvd}
\textsc{Unweighted Vertex Integrity} is fixed-parameter tractable parameterized by cluster vertex deletion number.
\end{theorem}

We split the proof of \cref{thm:unweighted_cvd} into two parts.
In the first part, we show that the problem can be solved by solving $2^{\bigOh(2^{k})}$ instances of a subproblem (\textsc{SubViCvd}) defined below,
where $k = \cvd(G)$.
In the second part, we present an algorithm that solves the subproblem in time $g(k) \cdot \textrm{poly}(|V(G)|)$ with some~$g$.

A vertex set $S$ is \emph{$(D,k)$-feasible} if $S \cap D = \emptyset$ and $|S| \le k$.
An \emph{$(D,k)$-feasible $\vi$-set} of a graph $G$ is a set $S \subseteq V(G)$
that minimizes $|S| + \max_{C \in \cc(G -S)}|V(C)|$ under the condition that $S$ is $(D,k)$-feasible.
Since making a set smaller never loses $(D,k)$-feasibility,
we can directly use \cref{lem:irredundant} to obtain the $(D,k)$-feasible counterparts of 
\cref{cor:irredundant}, \cref{obs:simplicial}, and \cref{lem:twin-all/nothing}.
Thus we can use them also for the $(D,k)$-feasible setting.

Now the subproblem is defined as follows.
\begin{description}
  \setlength{\itemsep}{0pt}
  \item[Problem.] \textsc{SubViCvd}
  \item[Input.] A graph $G$, an integer $k$, and a cluster vertex deletion set $D$ of $G$ with $\lvert \cc(G[D]) \rvert \le k$.
  \item[Output.] A $(D, k)$-feasible $\vi$-set of $G$.
  \item[Parameter.] $k$.
\end{description}
Note that, in an instance $(G,k,D)$ of \textsc{SubViCvd},
the size of $D$ is not necessarily bounded by a function of $k$.

\subsubsection{Reduction to nice instances of \textsc{SubViCvd}.}
\label{sec:reduction-to-SubViCvd}

Let $G$ be a graph and $D$ be a cluster vertex deletion set of $G$ with $\cvd(G) = k$.
We find a $\vi$-set $S$ of $G$ such that each twin class of $G$ is either completely contained in $S$ or has no intersection with $S$.
By \cref{lem:twin-all/nothing}, such a $\vi$-set exists.

Let $C$ be a connected component of $G - D$ that has the maximum number of vertices.
Since $C$ is a complete graph, the vertices of $C - S$ are included in the same connected component of $G - S$.
Thus, $\vi(G) \ge |S| + |V(C) \setminus S| = |S \setminus V(C)| + |V(C)|$.
On the other hand, as $D$ is a $\vi(|D| + |V(C)|)$-set, we have $\vi(G) \le |D| + |V(C)| = k + |V(C)|$.
Hence, we have
\begin{equation} 
  |S \setminus V(C)| \le  k.  \label{eq:small-vi-set}
\end{equation}

For each twin class $T \subseteq V(C)$ of $G$ with $|T| > k$, we guess whether $T \subseteq S$ or not.
Since $C$ is a complete graph and $N(V(C)) \subseteq D$, there are at most $2^{k}$ different twin classes in $C$,
and thus there are at most $2^{2^{k}}$ possible ways for this guess.
Let $S_{C}$ be the union of the twin classes $T \subseteq V(C)$ guessed as $T \subseteq S$.
Now, by \cref{eq:small-vi-set}, it suffices to find an $(\emptyset, k)$-feasible $\vi$-set of $G - S_{C}$.

We now guess which vertices of $D$ are included in $S$.
We call the set of guessed vertices $S_{D}$. There are at most $2^{k}$ possible options for this guess.
Note that $D \setminus S_{D}$ is a cluster vertex deletion set of $G - (S_{C} \cup S_{D})$.

Let $G' = G - (S_{C} \cup S_{D})$ and $D' = D \setminus S_{D}$.
Assuming that $S_{C}$ and $S_{D}$ are correctly guessed, 
the remaining problem is to find a $(D', k)$-feasible $\vi$-set of the graph $G'$.
That is, our task is to solve the instance $(G', k, D')$ of \textsc{SubViCvd}.

Before solving the instance of \textsc{SubViCvd}, we enlarge the cluster vertex deletion set.
Let $D'_{+}$ be the union of $D'$ and all twin classes $T$ in $G'-D'$ satisfying both $|T| > k$ and $N(T) \cap D' \ne \emptyset$.
Adding such large twin classes is safe as we are finding a set of size at most $k$.
To see that $(G', k, D'_{+})$ is an instance of \textsc{SubViCvd}, 
observe that $|{\cc(G'[D'_{+}])}| \le |{\cc(G'[D'])}| \le k$
since each vertex in $D'_{+} \setminus D'$ has neighbors in $D'$.

The discussion so far shows that it suffices to solve $2^{2^{k} + k}$ instances $(G',k,D'_{+})$ of \textsc{SubViCvd} 
such that each twin class $T$ of $G-D'_{+}$ satisfies $|T| \le k$ or $N(T) \cap D'_{+} = \emptyset$.
We call such an instance \emph{nice}.

\subsubsection{Solving a nice instance of \textsc{SubViCvd}.}
\label{sec:solving-SubViCvd}

Let $(G,k,D)$ be a nice instance of \textsc{SubViCvd}. (Notice that we renamed the objects.)

For a vertex $v \in V(G) \setminus D$, let $\mathcal{N}_{D}(v)$ be the set of connected components of $G[D]$ that include a neighbor of $v$;
that is, $\mathcal{N}_{D}(v) = \{C \in \cc(G[D]) \mid N(v) \cap V(C) \ne \emptyset\}$.
Since $D$ is a cluster vertex deletion set of $G$,
two vertices $v, v'$ with $\mathcal{N}_{D}(v) = \mathcal{N}_{D}(v')$ in the same connected component of $G - D$, which is a complete graph,
 play the same role in \textsc{SubViCvd}.
Namely, if a $(D,k)$-feasible $\vi$-set $S$ contains $v$ but not $v'$,
then $(S \setminus \{v\}) \cup \{v'\}$ is also a $(D,k)$-feasible $\vi$-set.

We say that $K, K' \in \cc(G - D)$ are of the \emph{same type} if
(1) for every non-empty $\mathcal{C} \subseteq \cc(G[D])$,
$|\{ v \in K \mid \mathcal{N}_{D}(v) = \mathcal{C} \}| = |\{ v \in K' \mid \mathcal{N}_{D}(v) = \mathcal{C} \}|$, and
(2) $\{v \in K  \mid \mathcal{N}_{D}(v) = \emptyset\}$ and 
$\{v \in K' \mid \mathcal{N}_{D}(v) = \emptyset\}$ are both empty or both non-empty.
The relation of having the same type is an equivalence relation among the connected components of $G - D$.
Observe that there are at most $2 (k+1)^{2^{k} - 1}$ equivalence classes (or \emph{types}) as
$|{\cc(G[D])}| \le k$ and every twin class of $G - D$ with neighbors in $D$ has size at most $k$.

\begin{myclaim}
\label{clm:larger-is-better}
There is a $(D, k)$-feasible $\vi$-set $S$ of $G$ satisfying the conditions
that (1) $N(S) \subseteq D$ and
(2) if $P$ and $Q$ are connected components of $G-D$ that have the same type and $|V(P)| < |V(Q)|$,
then $S$ intersects $P$ only if $S$ intersects $Q$ as well.
\end{myclaim}
\begin{subproof}
[Proof of \cref{clm:larger-is-better}]
Let $S$ be a $(D, k)$-feasible $\vi$-set of $G$ that contains no simplicial vertex.
(Recall that \cref{obs:simplicial} can be used also in the $(D, k)$-feasible setting.)
We can see that $N(S) \subseteq D$ since each vertex in $V(G) \setminus D$ is either a vertex with a neighbor in $D$ 
or a simplicial vertex. (Recall that  $D$ is a cluster vertex deletion set.)
We assume that, under this condition, $S$ maximizes $\beta(S) \coloneqq \sum_{C \in \cc(G-D), \, V(C) \cap S \ne \emptyset} |V(C)|$.
Let $P$ and $Q$ be connected components of $G-D$ that have the same type and $|V(P)| < |V(Q)|$.
Observe that $P$ and $Q$ both contain vertices with no neighbors in $D$ since they have the same type but different sizes.
Suppose to the contrary that $S$ intersects $P$ but not $Q$.

Since $S$ contains no simplicial vertex,
each vertex in $S \cap V(P)$ is adjacent to some vertices in $D$, and thus $V(P) \setminus S \ne \emptyset$.
Since $P$ and $Q$ have the same type, for every non-empty subset $\mathcal{C} \subseteq \cc(G[D])$, 
the numbers of vertices in $P$ and $Q$ having neighbors exactly in $\mathcal{C}$ are the same.
Thus there is an injection $f \colon S \cap V(P) \to V(Q)$ such that for each $v \in S \cap V(P)$, 
the neighbors of $v$ and $f(v)$ belong to the same set of connected components of $G[D]$.
Let $S'$ be the set obtained from $S$ by swapping $S \cap V(P)$ with $f(S \cap V(P))$; that is, $S' = (S \setminus V(P)) \cup f(S \cap V(P))$. 
Note that $|S| = |S'|$ and $|V(P) \setminus S| < |V(Q) \setminus S'|$.

Let $C_{P}$ and $C_{Q}$ be the (possibly identical) connected components of $G-S$
that contain $V(P) \setminus S$ ($\ne \emptyset$) and $V(Q) \setminus S$ ($=V(Q)$), respectively.
Similarly, let $C'_{P}$, $C'_{Q}$ be the connected components of $G-S'$
that contain $V(P) \setminus S'$ ($=V(P)$) and $V(Q) \setminus S'$ ($\ne \emptyset$), respectively.
Note that such components are well defined as $P$ and $Q$ are complete graphs and $V(P) \setminus S$ and $V(Q) \setminus S'$ are non-empty.
We show that
\begin{equation}
  \max\{|V(C_{P})|, |V(C_{Q})|\} \ge \max\{|V(C'_{P})|, |V(C'_{Q})|\}. \label{eq:cvd-large-comp}
\end{equation}
Observe that $C_{P} = C_{Q}$ if and only if $C'_{P} = C'_{Q}$,
and in such a case, they have the same number of vertices.
In the following, we assume that $C_{P} \ne C_{Q}$ (and thus $C'_{P} \ne C'_{Q}$).
This implies that $C_{Q}$ and $C'_{Q}$ have no intersection with $P$.
Let $p$ and $q$ be the number of vertices in $P$ and $Q$, respectively, that have no neighbors in $D$. Note that $p < q$.
Since $V(P) \setminus S$ and $V(Q) \setminus S'$ have the same adjacency to $\cc(G[D])$,
we have $|V(C_{P})| - p = |V(C'_{Q})| - q$ and $|V(C'_{P})| - p = |V(C_{Q})| - q$,
and thus $|V(C_{P})| < |V(C'_{Q})|$ and $|V(C'_{P})| < |V(C_{Q})|$.
We can also see that $V(C'_{Q}) \subseteq V(C_{Q})$ 
since $S \setminus V(P) = S' \setminus V(Q)$ and $C_{Q}$ and $C'_{Q}$ have no intersection with $P$.
Thus \cref{eq:cvd-large-comp} holds.

Observe that $\cc(G-S) \setminus \{C_{P}, C_{Q}\} = \cc(G-S') \setminus \{C'_{P}, C'_{Q}\}$.
This implies that $S'$ is also a $(D, k)$-feasible $\vi$-set.
This contradicts the assumption that $\beta(S)$ is maximum as $\beta(S') = \beta(S)-|V(P)|+|V(Q)| > \beta(S)$.
\end{subproof}

Let $S$ be a $(D,k)$-feasible $\vi$-set satisfying the conditions in \cref{clm:larger-is-better}.
Observe that the second condition guarantees that 
if $S$ intersects $k'$ ($\le k$) connected components $K_{1}, \dots, K_{k'}$ of $G-D$ with the same type,
then $K_{1}, \dots, K_{k'}$ are the largest ones in that type.
In other words, for each type of connected components of $G-D$,
only the $k$ largest ones (where ties are broken arbitrarily) can intersect $S$.
This observation leads to the following algorithm for solving \textsc{SubViCvd}
on the nice instance $(G,k,D)$.
\begin{enumerate}
  \item Select at most $k$ connected components of $G-D$.
  \begin{itemize}
    \item When selecting multiple components of the same type, pick them in non-decreasing order of their sizes (by breaking ties arbitrarily).
    \item There are at most $(2(k+1)^{2^{k}-1})^{k}$ options for this phase as there are at most $2(k+1)^{2^{k}-1}$ types.
  \end{itemize}
  \item From each of the selected connected components, select at most $k$ vertices (while keeping the total number of selected vertices to be at most $k$).
  \begin{itemize}
    \item The vertices of a connected component of $G-D$ with neighbors in $D$ can be partitioned into at most $2^{k}-1$ equivalence classes by $\mathcal{N}_{D}$.
    Thus there are at most $k^{2^{k}-1}$ options for one selected connected component.
    \item In total, there are at most $(k^{2^{k}-1})^{k}$ possible options for this phase.
  \end{itemize}
\end{enumerate}
The number of candidates for $(D,k)$-feasible $\vi$-set $S$ enumerated in the algorithm above
is bounded by a function $g(k)$ that depends only on~$k$.
Since they can be enumerated in time polynomial per candidate,
the algorithm runs in $g(k) \cdot \textrm{poly}(|V(G)|)$ time.

\subsection{The unary-weighted problem parameterized by $\cw$}
We now consider the most general parameter in this paper, clique-width~\cite{CourcelleMR00}.
We show that \textsc{Unary Weighted Vertex Integrity} belongs to $\mathrm{XP}$ parameterized by clique-width.

A vertex-labeled graph is a \emph{$c$-labeled graph} if each vertex has a label in $\{1,\dots,c\}$.
For $i \in \{1,\dots,c\}$, an \emph{$i$-vertex} in a $c$-labeled graph is a vertex of label $i$.
A \emph{$c$-expression} is a rooted binary tree such that:
\begin{itemize}
  \setlength{\itemsep}{0pt}
  \item each leaf has label $\circ_{i}$ for some $i \in \{1,\dots,c\}$;
  \item each non-leaf node with two children has label $\cup$;
  \item each non-leaf node with one child has label $\rho_{i,j}$ or $\eta_{i,j}$ with $i,j \in \{1,\dots,c\}$, $i \ne j$.
\end{itemize}
Each node in a $c$-expression represents a $c$-labeled graph as follows:
\begin{itemize}
  \setlength{\itemsep}{0pt}
  \item a $\circ_{i}$-node represents a one-vertex graph with one $i$-vertex;
  \item a $\cup$-node represents the disjoint union of the $c$-labeled graphs represented by its children;
  \item a $\rho_{i,j}$-node represents the $c$-labeled graph obtained from  its child by replacing the labels of all $i$-vertices with $j$;
  \item an $\eta_{i,j}$-node represents the $c$-labeled graph obtained from its child by adding all possible edges between the $i$-vertices and the $j$-vertices.
\end{itemize}
The \emph{clique-width} $\cw(G)$ of an unlabeled graph $G$ is the minimum integer $c$ such that
there is some $c$-expression that represents a $c$-labeled graph isomorphic to $G$ when the labels are disregarded.
It is known that for any constant $c$,
one can compute a $(2^{c+1}-1)$-expression of a graph of clique-width $c$ in $\bigOh(n^{3})$ time~\cite{HlinenyO08,OumS06,Oum08}.

\begin{theorem}
\label{thm:unary-cw-xp}
\textsc{Unary Weighted Vertex Integrity} belongs to $\mathrm{XP}$ parameterized by clique-width.
\end{theorem}
\begin{proof}
Let $G = (V,E, \weight)$ be a vertex-weighted graph.
Let $W = \weight(V)$.
Let $T$ be a $c$-expression of $G$. In the following, we consider $c$ as a fixed constant (e.g., in the $\bigOh(\cdot)$ notation).
For each node $t$ in $T$, let $G_{t}$ be the $c$-labeled graph represented by $t$ and $V_{t}$ be the vertex set of $G_{t}$.
For each $i \in \{1,\dots,c\}$, let $V_{t}^{i}$ denote the set of $i$-vertices in $G_{t}$.
For sets $S \subseteq V_{t}$ and $L \subseteq \{1,\dots,c\}$,
we denote by $\cc_{t}(S,L)$ the set of connected components of $G_{t}-S$ that include exactly the labels in $L$. That is,
\[
  \cc_{t}(S,L) =  \{C \in \cc(G_{t}-S) \mid L = \{\text{the label of $v$} \mid v \in  V(C) \}\}.
\]

For each node $t$ in $T$, we construct a table 
$\DP_{t}(\scom, \mcom) \in \{\true, \false\}$
with indices 
$\scom, \mcom \colon 2^{\{1,\dots,c\}} \to \{0,\dots,W\}$.
We set $\DP_{t}(\scom, \mcom) =  \true$
if and only if there exists a set $S \subseteq V_{t}$ such that, for all $L \subseteq \{1,\dots,c\}$,
\begin{align*}
  \scom(L) &= \sum_{C \in \cc_{t}(S,L)} \weight(V(C)), &
  \mcom(L) &= \max_{C \in \cc_{t}(S,L)} \weight(V(C)),
\end{align*}
assuming that $\max$ returns $0$ when applied to the empty set.
Note that the weight of a set $S$ corresponding to a tuple $(\scom, \mcom)$ can be uniquely determined, if any exists, as follows:
\begin{align*}
\weight(S) 
&= W - \sum_{C \in \cc(G_{t}-S)} \weight(V(C)) 
\\
&= W - \sum_{L \subseteq \{1,\dots,c\}} \, \sum_{C \in \cc_{t}(S,L)} \weight(V(C)) 
\\
&= W - \sum_{L \subseteq \{1,\dots,c\}} \scom(L).
\end{align*}
Let $r$ be the root node of $T$.
Observe that $\wvi(G)$ is the minimum integer $k$ such that
there exist $\scom$ and $\mcom$ satisfying that
$\DP_{r}(\scom, \mcom) =  \true$ and
$k = (W-\sum_{L \subseteq \{1,\dots,c\}} \scom(L)) + \max_{L \subseteq \{1,\dots,c\}} \mcom(L)$,
where $W-\sum_{L \subseteq \{1,\dots,c\}} \scom(L)$ corresponds to the weight of a $\wvi(k)$-set $S$
and $\max_{L \subseteq \{1,\dots,c\}} \mcom(L)$ is the maximum weight of a connected component of $G-S$.

We compute all entries $\DP_{t}(\scom, \mcom)$ in a bottom-up manner.
Note that there are $\bigOh(|V(T)| \cdot W^{2 \cdot 2^{c}})$ such entries.
Hence, to prove the theorem, it suffices to show that
each entry can be computed in time $\bigOh(W^{f(c)})$ for some computable function $f$
assuming that the entries for the children of $t$ are already computed.

For a leaf node $t$, $\DP_{t}(\scom, \mcom)$ can be computed in $\bigOh(1)$ time.
Let $\circ_{i}$ be the label of $t$.
We have $\DP_{t}(\scom, \mcom) = \true$ if and only if
\begin{itemize}
  \item $\scom(L) = \mcom(L) = 0$ for all $L \ne \{i\}$, and
  \item either $\scom(\{i\}) = \mcom(\{i\}) = \weight(V_{t}^{i})$ or $\scom(\{i\}) = \mcom(\{i\}) = 0$,
  where the first case corresponds to $S = \emptyset$ and the second one to $S = V_{t}^{i} = V_{t}$.
\end{itemize}
These conditions can be checked in $\bigOh(1)$ time.

For a $\cup$-node $t$, $\DP_{t}(\scom, \mcom)$ can be computed in $\bigOh(W^{3 \cdot 2^{c}})$ time as follows.
Let $t_{1}$ and $t_{2}$ be the children of $t$ in $T$.
Now, $\DP_{t}(\scom, \mcom) = \true$ if and only if there exist
tuples $(\scom_{1}, \mcom_{1})$ and $(\scom_{2}, \mcom_{2})$
such that
\begin{itemize}
  \item $\DP_{t_{1}}(\scom_{1}, \mcom_{1}) = \DP_{t_{2}}(\scom_{2}, \mcom_{2}) = \true$,
  \item $\scom(L) = \scom_{1}(L) + \scom_{2}(L)$ for all $L \subseteq \{1,\dots,c\}$,
  \item $\mcom(L) = \max\{\mcom_{1}(L), \mcom_{2}(L)\}$ for all $L \subseteq  \{1,\dots,c\}$.
\end{itemize}
Since $\scom_{2}$ is uniquely determined from $\scom_{1}$,
there are $\bigOh(W^{2^{c}})$ possible pairs for $(\scom_{1}, \scom_{2})$.
There are at most $\bigOh(W^{2 \cdot 2^{c}})$ candidates for $(\mcom_{1}, \mcom_{2})$.
In total, there are at most $W^{3 \cdot 2^{c}}$ candidates for the tuples.
The conditions for each candidate can be checked in $\bigOh(1)$ time, and thus the lemma holds.

For a $\rho_{i,j}$-node $t$, $\DP_{t}(\scom, \mcom)$ can be computed in $\bigOh(W^{2 \cdot 2^{c}})$ time as follows.
Let $t_{1}$ be the child of $t$ in $T$.
Observe that a connected component with label set~$L$ in $G_{t}$
has label set either $L$, $L \cup \{i\}$, or $L \cup \{i\} \setminus \{j\}$ in $G_{t_{1}}$
Thus $\DP_{t}(\scom, \mcom) = \true$
if and only if there exists
a tuple $(\scom_{1}, \mcom_{1})$ with 
$\DP_{t_{1}}(\scom_{1}, \mcom_{1}) = \true$ such that
for every $L \subseteq \{1,\dots,c\}$, the following holds:
\begin{itemize}
  \item if $i \in L$, then   $\scom(L) = 0$ and $\mcom(L) = 0$;

  \item if $i, j \notin L$, then
  $\scom(L) = \scom_{1}(L)$ and $\mcom(L) = \mcom_{1}(L)$;

  \item if $i \notin L$ and $j \in L$, then:
  \begin{itemize}
    \item $\scom(L) = \scom_{1}(L) + \scom_{1}(L \cup \{i\}) + \scom_{1}(L \cup \{i\} \setminus \{j\})$,
    \item $\mcom(L) = \max\{\mcom_{1}(L), \mcom_{1}(L \cup \{i\}), \mcom_{1}(L \cup \{i\} \setminus \{j\})\}$.
  \end{itemize}
\end{itemize}
The claimed running time follows from the facts that there are $\bigOh(W^{2 \cdot 2^{c}})$ candidates for 
$(\scom_{1}, \mcom_{1})$
and that each candidate can be checked in $\bigOh(1)$ time.

For an $\eta_{i,j}$-node $t$, $\DP_{t}(\scom, \mcom)$ can be computed in time $\bigOh(W^{2 \cdot 2^{c}})$ as follows.
Let $t_{1}$ be the child of $t$ in $T$.
For every $h \in \{1,\dots,c\}$, let $\scom(h) = \sum_{L \subseteq \{1,\dots,c\}, \, h \in L} \scom(L)$.
If there is a set $S \subseteq V(G_{t})$ corresponding to the tuple $(\scom, \mcom)$,
then for each $h \in \{1,\dots,c\}$, $V_{t}^{h} \setminus S$ is non-empty if and only if 
$\scom(h) \ne 0$.
If $V_{t}^{i} \setminus S \ne \emptyset$ and $V_{t}^{j} \setminus S \ne \emptyset$,
then all connected components in $G_{t_{1}} -S$ containing an $i$-vertex or a $j$-vertex
will be merged into a single connected component in $G_{t} - S$.
Otherwise, 
$V_{t}^{i} \subseteq S$ or $V_{t}^{j} \subseteq S$ holds,
and no connected components of $G_{t_{1}} - S$ will be merged in $G_{t} - S$ in this case.
Now we can see that 
$\DP_{t}(\scom, \mcom) = \true$
if and only if there exists
a tuple $(\scom_{1}, \mcom_{1})$ with $\DP_{t_{1}}(\scom_{1}, \mcom_{1}) = \true$
such that for each $L \subseteq \{1,\dots,c\}$, the following holds:
\begin{itemize}
  \item if $i,j \notin L$, then $\scom(L) = \scom_{1}(L)$ and $\mcom(L) = \mcom_{1}(L)$;

  \item if $|\{i,j\} \cap L| = 1$, then $\scom(L) = \scom_{1}(L)$, $\mcom(L) = \mcom_{1}(L)$,
  and at least one of $\scom(i)$ and $\scom(j)$ is $0$, 

  \item if $i,j \in L$, then $L =  \bigcup_{L' \in \mathcal{L}_{i,j}} L'$,
  $\scom(L) = \sum_{Q \in \mathcal{L}_{i,j}} \scom_{1}(Q)$, $\mcom(L) = \scom(L)$,
  and both $\scom(i)$ and $\scom(j)$ are positive,
  where $\mathcal{L}_{i,j} = \{L \subseteq \{1,\dots,c\} \mid L \cap \{i,j\} \ne \emptyset, \scom_{1}(L) > 0\}$.
\end{itemize}
The number of candidates for $(\scom_{1}, \mcom_{1})$
is $\bigOh(W^{2 \cdot 2^{c}})$ and each of them can be checked in $\bigOh(1)$ time.
\end{proof}


\section{Negative results}
In some of our hardness proofs, it is convenient to consider disconnected graphs
as we can easily argue that specific vertices are not selected into a $\wvi$-set.
However, we present the proofs for connected graphs to make the hardness results stronger.
In some cases, we need some careful treatment to make the graphs connected.
On the other hand, some other cases allow the following easy argument.
A \emph{universal vertex} in a graph is a vertex adjacent to every other vertex in the graph.
\begin{lemma}
\label{lem:universal}
Let $G = (V,E,\weight)$ be a vertex-weighted graph with a universal vertex $u$.
Then, $\wvi(G) = \wvi(G-u) + \weight(u)$.
\end{lemma}
\begin{proof}
We first show that $\wvi(G) \le \wvi(G-u) + \weight(u)$.
Let $S$ be a $\wvi$-set of $G-u$.
Since $G - (S \cup \{u\}) = (G-u) - S$ holds, we have
\begin{align*}
  \wvi(G)
  &\le
  \weight(S \cup \{u\}) + \max_{C \in \cc(G - (S \cup \{u\}))} \weight(V(C)) \\
  &=
  \weight(u)+\weight(S) + \max_{C \in \cc((G-u) - S)} \weight(V(C)) 
  =
  \weight(u)+\wvi(G-u).
\end{align*}

Next we show that $\wvi(G-u) \le \wvi(G) - \weight(u)$.
Let $S$ be a $\wvi$-set of $G$.
If $u \in S$, then $G - S = (G-u) - (S \setminus \{u\}))$ holds, and thus
\begin{align*}
  \wvi(G-u)
  &\le
  \weight(S \setminus \{u\}) + \max_{C \in \cc((G-u) - (S \setminus \{u\}))} \weight(V(C)) \\
  &=
  \weight(S) - \weight(u) + \max_{C \in \cc(G - S)} \weight(V(C)) 
  =
  \wvi(G) - \weight(u).
\end{align*}
If $u \notin S$, then $G-S$ is connected, and thus $\wvi(G) = \weight(V)$ holds.
This implies that $\wvi(G-u) \le \weight(V(G-u)) = \weight(V) - \weight(u) = \wvi(G) - \weight(u)$.
\end{proof}

\subsection{The unweighted problem parameterized by $\pw$}

Given a graph $G$, and integers $\ell$ and $p$, \textsc{Component Order Connectivity} asks 
whether the $\ell$-component order connectivity of $G$ is at most $p$.
The special case of \textsc{Component Order Connectivity} with $\ell = p$ is called \textsc{Fracture Number}.
(See \cref{sec:related-paras} for the definitions of $\ell$-component order connectivity and fracture number.)

To show the W[2]-hardness of \textsc{Unweighted Vertex Integrity} parameterized by pathwidth,
we first show that \textsc{Fracture Number} (and thus \textsc{Component Order Connectivity} as well) is W[2]-hard parameterized by pathwidth,
and then we present a pathwidth-preserving reduction from \textsc{Component Order Connectivity} to \textsc{Unweighted Vertex Integrity}.

Let $G = (V,E)$ be a connected graph.
A set $S \subseteq V$ is a \emph{connected safe set} if $G[S]$ is connected and $|S| \ge |V(C)|$ for each $C \in \cc(G-S)$.
Given a graph $G$ and an integer $k$, \textsc{Connected Safe Set} asks whether $G$ contains a connected safe set of size at most $k$.
Belmonte et al.~\cite{BelmonteHKLOO20} showed that 
\textsc{Connected Safe Set} is $\mathrm{W[2]}$-hard parameterized by pathwidth
even if the input graph $G$ contains a universal vertex.
This almost directly implies that \textsc{Fracture Number} is $\mathrm{W[2]}$-hard parameterized by pathwidth as we show below.
Here we omit the definition of pathwidth as it is not necessary.
\begin{proposition}
\label{prop:COC_pw_W[2]}
\textsc{Fracture Number} is $\mathrm{W[2]}$-hard parameterized by pathwidth.
\end{proposition}
\begin{proof}
Let $G = (V,E)$ a graph that contains a universal vertex $u$.
We show that $(G, k)$ is a yes-instance of \textsc{Connected Safe Set}
if and only if
$(G, k)$ is a yes-instance of \textsc{Fracture Number}.

To show the only-if direction, assume that $(G, k)$ is a yes-instance of \textsc{Connected Safe Set}.
A minimum connected safe set $S \subseteq V$ satisfies that $|S| \le k$ and $|V(C)| \le |S| \le k $ for each $C \in \cc(G-S)$.
Hence, $(G, k)$ is a yes-instance of \textsc{Fracture Number}.

To show the if direction, assume that $(G, k)$ is a yes-instance of \textsc{Fracture Number}.
Let $S \subseteq V$ be a set such that $|S| \le k$ and $|V(C)| \le k$ for each $C \in \cc(G-S)$.
We may assume without loss of generality that $|S| = k$ by adding arbitrary vertices if necessary.
If $G[S]$ is connected, then $S$ is indeed a connected safe set and we are done.
Assume that $S$ is not connected, and thus $S$ does not contain the universal vertex $u$.
Since $u$ is a universal vertex, $G-S$ is connected, which implies that $|V \setminus S| \le k$.
Now we construct a set $S'$ of size $k$ by first picking $u$ and then adding arbitrary $k-1$ vertices.
Since $S'$ contains $u$, the graph $G[S']$ is connected. 
As $|V \setminus S'| = |V \setminus S| \le k$, 
it holds for each $C \in \cc(G-S')$ that $|C| \le k = |S'|$.
Therefore, $S'$ is a connected safe set of size $k$.
\end{proof}

\begin{lemma}
\label{lem:reduction_COC->VI_unweighted}
There is a polynomial-time reduction from \textsc{Component Order Connectivity}
to \textsc{Unweighted Vertex Integrity} that increases pathwidth by at most~$1$.
\end{lemma}
\begin{proof}
Let $(H, \ell, p)$ be an instance of \textsc{Component Order Connectivity}.
We set $k = \ell p + \ell + p$.
We construct a graph $G$ by attaching $p$ pendants (i.e., vertices of degree~$1$) to each vertex of $H$
and then adding $k +1$ copies of $K_{1, k-p-1}$.
Note that $\pw(G) \le \pw(H)+1$ since removing the degree-$1$ vertices in $G$
decreases its pathwidth by at most~$1$ (see~\cite{BelmonteHKLOO20})
and $\pw(K_{1, k-p-1}) = 1$.
We show that $(G,k)$ is a yes-instance of \textsc{Unweighted Vertex Integrity} if and only if 
$(H, \ell, p)$ is a yes-instance of \textsc{Component Order Connectivity}.

To prove the if direction, assume that $(H, \ell, p)$ is a yes-instance of \textsc{Component Order Connectivity}
and $S \subseteq V(H)$ satisfies that $|S| \le p$ and $\max_{C \in \cc(H-S)} |V(C)| \le \ell$.
We show that $S$ is a $\vi(k)$-set of $G$. Since $|S| \le p$, it suffices to show that $\max_{C \in \cc(G-S)} |V(C)| \le k - p$.
If $C \in \cc(G-S)$ contains no vertex of $H$, then $C$ is one of the copies of $K_{1, k-p-1}$ or 
a single-vertex component corresponding to a pendant.
Otherwise, $V(C) \cap V(H)$ induces a connected component of $H - S$,
and thus $|V(C)| = (p+1) |V(C) \cap V(H)| \le (p+1) \ell = k - p$.

To prove the only-if direction, assume that $(G, k)$ is a yes-instance of \textsc{Unweighted Vertex Integrity}
and $S$ is an irredundant $\vi(k)$-set of $G$.
Since there are $k+1$ copies of $K_{1, k-p-1}$,
there is one that does not intersect $S$.
This implies that $\max_{C \in \cc(G-S)} |V(C)| \ge k-p$, and thus $|S| \le p$.
Let $S' = S \cap V(H)$ and 
let $C' \in \cc(H - S')$ be an arbitrary connected component of $H - S'$.
As $|S'| \le p$, it suffices to show that $|V(C')| \le \ell$.
Let $C$ be the connected component of $G - S$ such that $V(C') \subseteq V(C)$.
Observe that $S$ contains no vertices of degree~$1$ since such vertices are redundant.
This implies that, for each vertex of $C'$, $C$ contains all $p$ pendants attached to it,
and thus, $|V(C)| \ge (p+1) |V(C')|$.
Since $|V(C)| \le k = \ell p + \ell + p < (p+1)(\ell+1)$,
we have $|V(C')| < \ell+1$.
\end{proof}

Since adding a universal vertex increases pathwidth by $1$,
\cref{lem:reduction_COC->VI_unweighted,prop:COC_pw_W[2],lem:universal} imply the following.
\begin{corollary}
\label{cor:unweighted_pw}
\textsc{Unweighted Vertex Integrity} on connected graphs is $\mathrm{W[2]}$-hard parameterized by pathwidth.
\end{corollary}

\subsection{The unary-weighted problem parameterized by $\cvd$, $\fvs$, $\vi$}
We now show that once we allow unary vertex weights, the problem becomes intractable 
even if the unweighted vertex integrity of the input graph is small.

\begin{theorem}
\label{thm:unary-vi-cvd-fvs_Wh}
\textsc{Unary Weighted Vertex Integrity} on connected graphs is $\mathrm{W[1]}$-hard parameterized 
by cluster vertex deletion number plus unweighted vertex integrity
or by feedback vertex number plus unweighted vertex integrity.
\end{theorem}
\begin{proof}
We show the hardness by presenting a reduction from \textsc{Unary Bin Packing}.
Given positive integers $t$ and $a_{1}, \dots, a_{n}$ in unary such that $\sum_{1 \le i \le n} a_{i} = t B$,
\textsc{Unary Bin Packing} asks whether there is a partition $(I_{1},\dots,I_{t})$ of $\{1,\dots,n\}$
such that $\sum_{i \in I_{j}} a_{i} = B$ for each $j \in \{1,\dots,t\}$.
\textsc{Unary Bin Packing} is W[1]-hard parameterized by $t$~\cite{JansenKMS13}.
We may assume without loss of generality that $a_{i} - 1 > (t-1)n$ for each $a_{i}$ 
by multiplying a large number to every $a_{i}$ if necessary.
We also assume that $\max_{i} a_{i} \le B$ since otherwise the instance is a no-instance.

We set $k = (t-1)n + 3B$ and construct a graph $G = (V,E)$ as follows. The vertex set is 
$V = \{u_{i} \mid 1 \le i \le n\} \cup \{v_{i}^{j} \mid 1 \le i \le n, \, 1 \le j \le t\} \cup \{w_{i} \mid 1 \le i \le t\} \cup \{x\}$.
For $1 \le i \le n$ and $1 \le j \le t$, the edge set $E$ contains the edge $\{u_{i}, v_{i}^{j}\}$ and $\{v_{i}^{j}, w_{j}\}$.
Additionally, for $1 \le i \le n$, each set $\{v_{i}^{1}, \dots, v_{i}^{t}\}$ may induce an arbitrary graph.
(We only need the cases where all of them are independent sets or cliques.)
We set $\weight(u_{i}) = a_{i}-1$,
$\weight(v_{i}^{j}) = 1$, and
$\weight(w_{j}) = 2B$ for all $1 \le i \le n$ and $1 \le j \le t$,
and $\weight(x) = 3B$.
See \cref{fig:unary-vi-cvd-fvs_Wh}.

\begin{figure}[ht]
\centering
\includegraphics[scale=1.0]{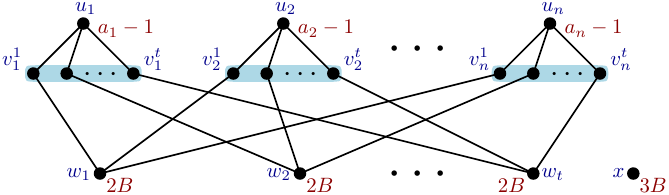}
\caption{The reduction in \cref{thm:unary-vi-cvd-fvs_Wh}.}
\label{fig:unary-vi-cvd-fvs_Wh}
\end{figure}

Let $W = \{w_{1}, \dots, w_{t}\}$.
We can see that $W$ is a $\vi(2t+1)$-set in the unweighted graph obtained from $G$ by ignoring the weight-function $\weight$.
If each $\{v_{i}^{1}, \dots, v_{i}^{t}\}$ is a clique,
then $W$ is a cluster vertex deletion set.
On the other hand, if each $\{v_{i}^{1}, \dots, v_{i}^{t}\}$ is an independent set,
then $W$ is a feedback vertex set.
Note that $G$ is not connected.
To make it connected, we can just add a universal vertex of weight $1$ and add $1$ to $k$.
The parameters increase at most by~$1$, and \cref{lem:universal} ensures the equivalence.
Therefore, it suffices to show that 
$(G,k)$ is a yes-instance of \textsc{Unary Weighted Vertex Integrity} 
if and only if 
$(t, a_{1}, \dots, a_{n})$ is a yes-instance of \textsc{Unary Bin Packing}.

To show the if direction, assume that 
$(t, a_{1}, \dots, a_{n})$ is a yes-instance of \textsc{Unary Bin Packing}.
Let $(I_{1},\dots,I_{t})$ be a partition of $\{1,\dots,n\}$
such that $\sum_{i \in I_{j}} a_{i} = B$ for each $j \in \{1,\dots,t\}$.
Let $S = \{v_{i}^{j} \mid i \notin I_{j}\}$. Note that $\weight(S) = (t-1)n$.
It suffices to show that each connected component of $G-S$ has weight at most $k-(t-1)n = 3B$.
The vertex $x$ forms a single-vertex connected component with weight $\weight(x) = 3B$.
Let $C \in \cc(G-S)$ be a connected component not containing~$x$.
Observe that in $G-S$, each $u_{i}$ is of degree~$1$, each $v_{i}^{j}$ with $i \in I_{j}$ has degree~$2$,
and there are no paths connecting $w_{j}$ and $w_{j'}$ for $j \ne j'$.
This implies that there exists a unique index $j \in \{1,\dots,t\}$
such that $V(C) = \{w_{j}\} \cup \{u_{i}, v_{i}^{j} \mid i \in I_{j}\}$,
and thus
$\weight(V(C)) = \weight(w_{j}) + \sum_{i \in I_{j}} (\weight(u_{i}) + \weight(v_{i}^{j})) = 2B + \sum_{i \in I_{j}} a_{i} = 3B$.

To show the only-if direction, 
assume that $(G,k)$ is a yes-instance of \textsc{Unary Weighted Vertex Integrity}
and let $S \subseteq V(G)$ be an irredundant $\wvi$-set.
By \cref{obs:simplicial}, $x \notin S$.
This implies that $\max_{C \in \cc(G-S)} \weight(V(C)) \ge \weight(x) = 3B$, and hence $\weight(S) \le k - 3B = (t-1)n$.
Since $\weight(u_{i}) = a_{i} - 1 > (t-1)n$, we have $u_{i} \notin S$ for all $i$.
Also, $w_{j} \notin S$ for all $j$ as $\weight(w_{j}) = 2B > \max_{i} a_{i} > (t-1)n$.
Thus, $S \subseteq \{v_{i}^{j} \mid 1 \le i \le n, \, 1 \le j \le t\}$ holds.
If $v_{i}^{j}, v_{i}^{j'} \notin S$ for some $i$ and $j \ne j'$,
then $u_{i}, v_{i}^{j}, v_{i}^{j'}, w_{j}, w_{j'}$ belong to the same connected component in $G-S$,
having weight more than $\weight(\{w_{j}, w_{j'}\}) = 4B > (t-1)n + 3B  = k$.
Hence, for each $i$, there is at most one index $j$ such that $v_{i}^{j} \notin S$.
Actually, since $\weight(S) \le (t-1)n$, there is exactly one such index for each $i$. Note that this implies that $\weight(S) = (t-1)n$.
Now we can define the partition $(I_{1},\dots,I_{t})$ of $\{1,\dots,n\}$ as $I_{j} = \{i \mid v_{i}^{j} \notin S\}$.
For $j \in \{1,\dots,t\}$, let $C_{j}$ be the connected component of $G-S$ that contains $w_{j}$.
We can see as before that $V(C_{j}) = \{w_{j}\} \cup \{u_{i}, v_{i}^{j} \mid i \in I_{j}\}$.
Since $\weight(V(C_{j})) \le k - \weight(S) = 3B$ and $\weight(w_{j}) = 2B$, we have
\[
 \sum_{i \in I_{j}} a_{i} = \weight(V(C_{j})) - \weight(w_{j}) \le B.
\]
Since the total weight is $\sum_{1 \le i \le n} a_{i} = t B$, 
we can conclude that $\sum_{i \in I_{j}} a_{i} = B$ for each $j \in \{1,\dots,t\}$.
\end{proof}

\subsection{The binary-weighted problem on subdivided stars}
A graph is a \emph{subdivided star} if it can be obtained from a star $K_{1,s}$ by subdividing each edge once.
Observe that the center of a subdivided star forms a cluster deletion set and a $\vi(3)$-set.
Thus a hardness result on subdivided stars immediately implies the same hardness with parameters $\cvd$ and $\vi$.

\begin{theorem}
\label{thm:binary_subdivided_star_NPc}
\textsc{Binary Weighted Vertex Integrity} is $\mathrm{NP}$-complete on subdivided stars.
\end{theorem}
\begin{corollary}
\label{cor:binary-vi-cvd-fvs_NPc}
\textsc{Binary Weighted Vertex Integrity} is $\mathrm{NP}$-complete on connected graphs
with cluster vertex deletion number~$1$, vertex integrity~$3$, and feedback vertex number~$0$. 
\end{corollary}
\begin{proof}
[Proof of \cref{thm:binary_subdivided_star_NPc}]
Clearly, the problem belongs to NP\@.
To show the NP-hardness, we give a reduction from \textsc{Partition}, which is (weakly) NP-complete~\cite{GareyJ79}.
Given $n$ positive integers $a_{1}, \dots, a_{n}$ with $\sum_{1 \le i \le n} a_{i} = 2B$,
\textsc{Partition} asks whether there is a set of indices $I \subseteq \{1,\dots,n\}$ such that $\sum_{i \in I} a_{i} = B$. 
We may assume that $a_{i} \le B$ for all $i$ since otherwise such $I$ cannot exist.

We set $k = (B+1)(B+2)$ and construct a graph $G = (V,E)$ as
\begin{align*}
  V &= \{r\} \cup \{u_{i}, v_{i} \mid 1 \le i \le n\} \cup \{w, x\}, \\
  E &= \{\{r, u_{i}\}, \{u_{i}, v_{i}\} \mid 1 \le i \le n\} \cup \{\{r, w\}, \{w,x\}\}.
\end{align*}
We set $\weight(r) = B+1$, 
$\weight(u_{i}) = a_{i}$ and $\weight(v_{i}) = a_{i} B$ for $1 \le i \le n$,
$\weight(w) = 1$, and $\weight(x) = (B+1)^{2} = k-(B+1)$.
See \cref{fig:binary_subdivided_star_NPc}.
We show that $(G, k)$ is a yes-instance of \textsc{Binary Weighted Vertex Integrity}
if and only if 
$(a_{1},\dots,a_{n})$ is a yes-instance of \textsc{Partition}.

\begin{figure}[ht]
\centering
\includegraphics[scale=1.0]{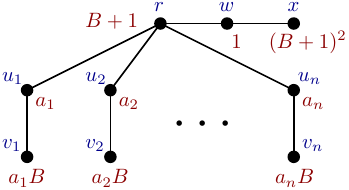}
\caption{The reduction in \cref{thm:binary_subdivided_star_NPc}.}
\label{fig:binary_subdivided_star_NPc}
\end{figure}

To show the if direction, assume that $(a_{1},\dots,a_{n})$ is a yes-instance of \textsc{Partition}.
Let $I \subseteq \{1,\dots,n\}$ be a set of indices such that $\sum_{i \in I} a_{i} = B$. 
We set $S = \{u_{i} \mid i \in I\} \cup \{w\}$ and show that $S$ is a $\wvi(k)$-set of $G$.
Since $\weight(S) = 1 + \sum_{i \in I} a_{i} = B + 1$,
it suffices to show that each connected component of $G-S$ has weight at most $k - (B+1) = (B+1)^{2}$.
This is true for the connected components of $G-S$ with only one vertex, $x$ or $v_{i}$,
as their weight is $(B+1)^{2}$ or $a_{i} B$ ($\le B^{2}$).
The weight of the other connected component of $G-S$, which contains $r$, is
$\weight(r) + (B+1) \sum_{i \notin I} a_{i} = (B+1) + (B+1) B = (B+1)^{2}$.

To show the only-if direction, assume that $(G, k)$ is a yes-instance of \textsc{Binary Weighted Vertex Integrity}.
Let $S$ be an irredundant $\wvi(k)$-set of $G$.
By \cref{obs:simplicial}, $S \cap \{x, v_{1}, \dots, v_{n}\} = \emptyset$.
Since $x \notin S$, it holds that $\weight(S) \le k - \weight(x) = B+1$.
If $r \in S$, then $S = \{r\}$ as $\weight(r) = B+1$, and thus $\cc(G-S)$ contains a connected component induced by $\{w,x\}$, 
which has weight $(B+1)^{2} + 1 = k-B$. Thus we have $r \notin S$.
Since $r, x \notin S$, $w \in S$ holds; otherwise, we have a connected component of weight at least $\weight(\{r,w,x\}) = k+1$. 
Now we know that $\{w \} \subseteq S \subseteq \{w, u_{1}, \dots, u_{n}\}$.
Let $I = \{i \mid u_{i} \in S\}$.
Since $\weight(S) \le B+1$, we have $\sum_{i \in I} a_{i} \le B$.
Suppose to the contrary that $\sum_{i \in I} a_{i} \le B - 1$.
Let $C$ be the connected component of $G-S$ that includes $r$.
Then,
\begin{align*}
  \weight(C) 
  &= 
  \weight(r) + (B+1)\sum_{i \notin I} a_{i} 
  \ge
  (B+1) + (B+1) (2B - (B-1))
  =
  (B+1)(B+2)
  = k.
\end{align*}
This contradicts the assumption that $S$ is a $\wvi(k)$-set as $\weight(S) \ge \weight(w) \ge 1$.
\end{proof}

\subsection{Graph classes}

Here we give additional results that address the complexity of \textsc{Unweighted Vertex Integrity}
on special classes of graphs, namely planar bipartite graphs and line graphs.
\begin{theorem}
\label{thm:planar-bipartite}
\textsc{Unweighted Vertex Integrity} is $\mathrm{NP}$-complete
on connected planar bipartite graphs of maximum degree~$4$.
\end{theorem}
\begin{proof}
The problem clearly belongs to NP\@.
To show the NP-hardness, we present a reduction from \textsc{Vertex Cover}.
Given a graph $H$ and an integer $p$, \textsc{Vertex Cover} asks whether $H$ has a vertex cover of size at most $p$.
It is known that \textsc{Vertex Cover} is NP-complete on connected cubic planar graphs~\cite{Mohar01}.

Let $(H,p)$ be an instance of \textsc{Vertex Cover}, where $H = (V,E)$ is a connected cubic planar graph.
We assume that $|V| \ge 6$ and $p \le |V|-1$ as otherwise the problem can be solved in polynomial time.
We set $k = 3p+10$ and construct a graph $G$ as follows.
We start with $H$ and then subdivide each edge once.
We denote by $V_{E}$ the new vertices introduced by the subdivisions.
To each original vertex $v \in V(H)$, we attach a path of $p+2$ vertices.
That is, the degree~$1$ vertex in the attached path is of distance $p+2$ from $v$.
We call the graph constructed so far $H'$.
Now we add $p+4$ connected components, where each of them is obtained from $K_{1,3}$ by subdividing each edge $p+5$ times
so that the center is of distance $p+6$ from the leaves.
We call these components the \emph{spiders} and denote by $R$ the set of centers of the spiders.
Finally, to make the graph connected, we fix an arbitrary injection $f \colon R \to V_{E}$
and add the edge $\{r, f(r)\}$ for every $r \in R$.
Note that such an injection $f$ exists as $|V_{E}| = |E| = 3|V|/2 \ge |V|+3 \ge p+4 = |R|$ as $p \le |V|-1$ and $|V| \ge 6$.
The constructed graph $G$ is connected, planar, bipartite, and of maximum degree~$4$. 
See \cref{fig:planar-bipartite}.

\begin{figure}[ht]
\centering
\includegraphics[scale=1.0]{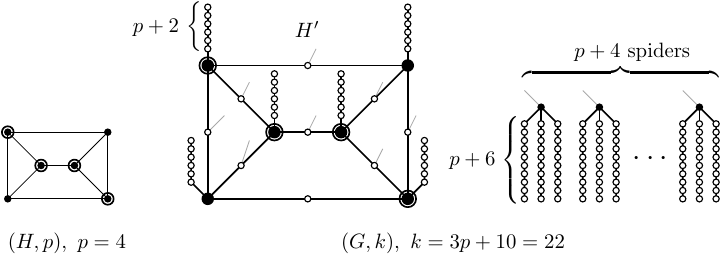}
\caption{The reduction in \cref{thm:planar-bipartite}. A size-$p$ vertex cover of $H$ is marked with circles.}
\label{fig:planar-bipartite}
\end{figure}

We show that $(G,k)$ is a yes-instance of \textsc{Unweighted Vertex Integrity}
if and only if $(H,p)$ is a yes-instance of \textsc{Vertex Cover}.

To show the if direction, assume that $(H,p)$ is a yes-instance of \textsc{Vertex Cover}.
Let $T \subseteq V$ be a vertex cover of $H$ such that $|T| \le p$.
Let $S = T \cup R$. Note that $|S| \le 2p+4$.
It suffices to show that $|V(C)| \le k - (2p+4) = p+6$ for every $C \in \cc(G-S)$.
If $C$ does not contain any vertex of $H$, then $C$ is either
the path of $p + 2$ vertices attached to a vertex of $H$,
the singleton component formed by a vertex in $V_{E}$, or 
one of the paths of $p+6$ vertices in a spider.
Assume that $C$ contains a vertex $v$ of $H$.
Since $T$ is a vertex cover of $H$, no neighbor of $v$ in $H$ is included in $C$.
Thus, $C$ consists of $v$, the path of $p+2$ vertices attached to $v$,
and the three vertices in $V_{E}$ corresponding to the edges of $H$ incident to $v$.
Hence, $|V(C)| =  1 + (p+2) + 3 = p+6$ holds.

To show the only-if direction,
assume that $(G,k)$ is a yes-instance of \textsc{Unweighted Vertex Integrity}.
Let $S$ be a $\vi(k)$-set of $G$.
First observe that $S$ intersects all spiders as each spider contains $3(p+6)+1 > k$ vertices.
Hence, $|S| \ge p+4$.
On the other hand, since $k/3 < p+4$, there is a spider that contains at most two vertices of $S$.
In particular, one of the three $(p+6)$-vertex paths in the spider does not intersect $S$.
This implies that $\max_{C \in \cc(G-S)} |V(C)| \ge p+6$, and thus $|S| \le k-(p+6) = 2p+4$.
Since the spiders contain at least $p+4$ vertices, $|S \cap V(H')| \le p$ holds.
We construct a set $T \subseteq V$ of vertices of $H$ as follows.
\begin{itemize}
  \item If $v \in S \cap V(H')$ is a vertex of $H$, then put $v$ into $T$.
  \item If $v \in S \cap V(H')$ is a vertex in the $(p+2)$-vertex path attached to a vertex $u$ of $H$, then put $u$ into $T$.
  \item If $v \in S \cap V(H')$ is the vertex in $V_{E}$ that corresponds to an edge $\{u,w\} \in E$, then put either one of $u$ and $v$ into $T$.
\end{itemize}
Clearly, $|T| \le |S \cap V(H')| \le p$.
We now show that $T$ is a vertex cover of $H$.
Suppose to the contrary that $H-T$ contains an edge $\{u,v\}$.
This implies that $S$ does not contain $u$, $v$, the vertex corresponding to the edge $\{u,v\}$ in $H$, and
the vertices in the paths attached to $u$ and $v$.
Since they induce a connected subgraph of $G$, 
there is a connected component $C$ of $G-S$ that contains all of these $3+2(p+2) = 2p+7$ vertices.
Since $|S| \ge p+4$, it holds that $|S| + |V(C)| \ge (p+4)+(2p+7) = 3p+11 > k$.
This contradicts the assumption that $S$ is a $\vi(k)$-set of $G$.
\end{proof}

For a graph $G = (V,E)$, its \emph{line graph} $L(G)$ is defined by setting 
$V(L(G)) = E$ and $E(L(G)) = \{\{e,f\} \subseteq E \mid e \cap f \ne \emptyset\}$.
A graph $H$ is a \emph{line graph} if there is a graph $G$ such that $L(G)$ is isomorphic to $H$.

We show the NP-completeness of \textsc{Unweighted Vertex Integrity} on connected line graphs.
To this end, we first define an intermediate problem.\footnote{%
Note that this problem, \textsc{Line Integrity}, is different from \textsc{Edge Integrity}~\cite{BaggaBGLP92},
which asks to remove a small number of edges to make the maximum number of \emph{vertices} in a connected component small.}
Given a graph $G$ and an integer $k$, \textsc{Line Integrity} asks whether there is $F \subseteq E(G)$ such that 
\[
  |F| + \max_{C \in \cc(G-F)} |E(C)| \le k.
\]
The set $F$ satisfying the inequality above is an \emph{$\li(k)$-set}.
We can see that \textsc{Line Integrity} is equivalent to \textsc{Unweighted Vertex Integrity} on line graphs.
\begin{lemma}
\label{lem:li-vi}
$F \subseteq E(G)$ is an $\li(k)$-set of $G$
if and only if $F$ is a $\vi(k)$-set of $L(G)$.
\end{lemma}
\begin{proof}
Let $F \subseteq E(G)$. Observe that,
for each $C \in \cc(G-F)$ with $E(C) \ne \emptyset$,
the line graph $L(G-F)$ contains $L(C)$ as a connected component.
Since $L(G) - F = L(G-F)$, we have
\[
\max_{C \in \cc(G-F)} |E(C)|
= \max_{C \in \cc(L(G)-F)} |V(C)|,
\]
and thus the claim holds.
\end{proof}

\begin{theorem}
\label{thm:line-integrity}
\textsc{Line Integrity} is $\mathrm{NP}$-complete on connected graphs.
\end{theorem}
\begin{proof}
The problem belongs to NP\@. To show the NP-hardness
we present a reduction from \textsc{Unary Bin Packing} (see \cref{thm:unary-vi-cvd-fvs_Wh}),
which is NP-complete when the number of bins $t$ is part of input~\cite{GareyJ79}.

Let $(t, a_{1}, \dots, a_{n})$ be an instance of \textsc{Unary Bin Packing}.
We assume that $a_{i} > tn$ for every $i$.
We set $k = (t-1)n + 3B + 1$ and construct a graph $G$ as follows.
We first add to $G$ sets of vertices $U = \{u_{1}, \dots, u_{n}\}$ and $V = \{v_{1}, \dots, v_{t}\}$,
and add all possible edges between $U$ and $V$. We denote by $e_{i,j}$ the edge between $u_{i}$ and $v_{j}$.
For $u_{i}$ and $v_{i}$, we attach $a_{i}-1$ and $2B$ pendant vertices, respectively.
We denote by $H$ the graph induced by $U$, $V$, and the pendants attached to them.
We further add a star $K_{1,3B}$ and call the center of the star $w$.
Finally, to make the graph connected, we add an edge between $w$ and $u_{n}$.
See \cref{fig:line-graph}.

\begin{figure}[ht]
\centering
\includegraphics[scale=1.0]{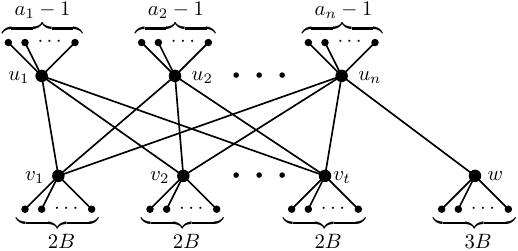}
\caption{The reduction in \cref{thm:line-integrity}.}
\label{fig:line-graph}
\end{figure}

We show that
$(G,k)$ is a yes-instance of \textsc{Line Integrity} 
if and only if 
$(t, a_{1}, \dots, a_{n})$ is a yes-instance of \textsc{Unary Bin Packing}.
(The proof is similar to the one for \cref{thm:unary-vi-cvd-fvs_Wh}.)

To show the if direction, assume that 
$(t, a_{1}, \dots, a_{n})$ is a yes-instance of \textsc{Unary Bin Packing}.
Let $(I_{1},\dots,I_{t})$ be a partition of $\{1,\dots,n\}$
such that $\sum_{i \in I_{j}} a_{i} = B$ for each $j \in \{1,\dots,t\}$.
Let $F = \{e_{i, j} \mid i \notin I_{j}\} \cup \{\{w,u_{n}\}\}$. 
Since $|F| = (t-1)n + 1$,
it suffices to show that each connected component of $G-F$ contains at most $k - |F| = 3B $ edges.
The connected component of $G-F$ containing $w$ has $3B$ edges.
Let $C \in \cc(G-F)$ be a connected component with $w \notin V(C)$.
By the definition of $F$, there is a unique $j \in \{1,\dots,t\}$
such that $V(C)$ consists of $v_{j}$, the vertices $u_{i}$ with $i \in I_{j}$, and the pendants attached to them.
Thus, $|E(C)| = 2B + \sum_{i \in I_{j}} a_{i} = 3B$.

To show the only-if direction, 
assume that $(G,k)$ is a yes-instance of \textsc{Line Integrity}
and let $F \subseteq E(G)$ be an $\li(k)$-set of $G$.
By \cref{lem:irredundant}, we may assume that $F$ is irredundant in $L(G)$.
In particular, we assume that $F$ does not include any pendant edge of $G$ since a pendant edge of $G$ is a simplicial vertex of $L(G)$.
This implies that $\max_{C \in \cc(G-F)} |E(C)| \ge 3B$ as $w$ is incident to $3B$ pendant edges.
Thus we have $|F| \le (t-1)n + 1$.
If $\{w, u_{n}\}\notin F$, then the component containing $w$ and $u_{n} $ contains $3B+a_{n} > 3B + tn > k$ edges.
Thus, $\{w, u_{n}\} \in F$ holds.
Let $F' = F \setminus \{\{w, u_{n}\}\}$.
We now know that $|F'| \le (t-1)n$ and $F' \subseteq \{e_{i,j} \mid 1 \le i \le n, \, 1 \le j \le t\}$.

If $e_{i,j}, e_{i,j'} \notin F'$ for some $i$ and $j \ne j'$,
then the connected component $C$ in $G-F$ containing $u_{i}$, $v_{j}$, $v_{j'}$, and the pendant vertices attached 
contains at least $a_{i} + 4B + 1$ edges, which is more than $k$ as $a_{i} > tn$.
Hence, for each $i$, there is at most one index $j$ such that $e_{i,j} \notin F'$.
Indeed, since $|F'| \le (t-1)n$, there is exactly one such index for each $i$. Note that this implies that $|F'| = (t-1)n$.

Now we can define the partition $(I_{1},\dots,I_{t})$ of $\{1,\dots,n\}$ as $I_{j} = \{i \mid v_{i}^{j} \notin F'\}$.
For $j \in \{1,\dots,t\}$, let $C_{j}$ be the connected component of $G - F$ that contains $v_{j}$.
Then, $C_{j}$ contains the vertices $u_{i}$ with $i \in I_{j}$.
As $C_{j}$ consists additionally with all pendants attached to them as well,
$|E(C_{j})| = 2B + \sum_{i \in I_{j}} a_{i}$.
Since $|E(C_{j})| \le k - |F| = 3B$ and $\weight(w_{j}) = 2B$, it holds that $\sum_{i \in I_{j}} a_{i} \le B$.
Considering the total weight $\sum_{1 \le i \le n} a_{i} = t B$,
we can see that $\sum_{i \in I_{j}} a_{i} = B$ for each $j \in \{1,\dots, t\}$.
\end{proof}
 
\cref{lem:li-vi,thm:line-integrity} immediately imply the following result.
\begin{corollary}
\label{cor:line-graph}
\textsc{Unweighted Vertex Integrity} is $\mathrm{NP}$-complete
on connected line graphs.
\end{corollary}


\section{An $\bigOh(\log \mathsf{opt})$-factor approximation algorithm}

In this section, we show the following theorem stating that the weighted vertex integrity can be approximated within an $\bigOh(\log \mathsf{opt})$ factor.
\begin{theorem}
\label{thm:log-approx}
There is a polynomial-time algorithm that, given a weighted graph $G = (V, E, \weight)$ with $\wvi(G) = k$,
outputs a $\wvi(k')$-set with $k' \in \bigOh(k \log k)$.
\end{theorem}

\cref{thm:log-approx} can be proved by slightly modifying an algorithm by Lee~\cite{Lee19}.
The problem studied in~\cite{Lee19} can be seen as a \emph{half-weighted} variant of $\ell$-component order connectivity:
given a weighted graph $G = (V,E,\weight)$ and an integer $\ell \in \mathbb{Z}^{+}$,
find a set $S \subseteq V$ with the minimum weight $\weight(S)$ such that each connected component of $G - S$ has at most $\ell$ vertices.
For this problem, Lee~\cite[Theorem~1]{Lee19} presented a bicriteria approximation algorithm that
runs in polynomial time and
outputs $S \subseteq V$ such that
\begin{itemize}
  \item each connected component of $G-S$ has at most $2\ell$ vertices and 
  \item $w(S) \in \bigOh(\mathsf{opt} \cdot \log \ell)$, where $\mathsf{opt}$ is the optimum value.
\end{itemize}
The algorithm has two phases: first it solves an LP relaxation of the problem
and then it rounds the obtained LP solution.
In the LP relaxation, $\ell$ is used as a fixed constant that bounds the maximum order of a connected component in $G-S$.\footnote{
In~\cite{Lee19}, the letter $k$ is used instead of $\ell$.}

To obtain an algorithm that satisfies the conditions in \cref{thm:log-approx},
we modify the LP by considering $\ell$ as a variable that bounds the maximum \emph{weight} of a connected component in $G-S$
and add $\ell$ to the objective function to be minimized.
Making these changes is straightforward and the proof in \cite{Lee19} works almost as is.
We omit the proof to avoid repeating almost the same argument.


\section{Concluding remarks}

We initiated the first systematic study of the problem of computing unweighted and weighted vertex integrity of graphs
in terms of structural graph parameters. Our results show sharp complexity contrasts of the problem.
We also obtained some hardness results on graph classes. See \cref{fig:graph-parameters,fig:graph-classes}.

There are still some cases where the complexity of the problem is unknown.
For example, one may ask the following questions.
\begin{itemize}
  \item What is the complexity of \textsc{Unweighted Vertex Integrity} parameterized by treedepth or by feedback vertex set number?
  \item What is the complexity of \textsc{Binary Vertex Integrity} parameterized by modular-width?
\end{itemize}

\subsection*{Additional remark}
Recently, improvements of some of our results and answers to the questions above have been announced in \cite{HanakaLVY24arxiv}.



\bibliographystyle{plainurl}
\bibliography{ref}

\end{document}